\documentclass[twocolumn,amsthm]{autart}%

\usepackage{url}
\usepackage{IEEEtrantools}

\usepackage{algorithm}
\usepackage{algorithmic}
\usepackage{graphicx}
\usepackage{epsfig}
\usepackage{amsmath}
\usepackage{amssymb}
\usepackage{mathrsfs}
\usepackage{amsfonts}
\usepackage{dsfont}
\usepackage{graphicx}
\usepackage{subcaption}
\usepackage{mathrsfs}
\usepackage{wrapfig}
\usepackage{hyperref}
\usepackage{xcolor}

\hypersetup{
    colorlinks=true,       
    linkcolor={blue!50!cyan},  
    citecolor={blue!50!cyan},  
    urlcolor={blue!50!cyan}    
}

\usepackage[round]{natbib}%
\theoremstyle{definition}

\newtheorem{definition}{Definition}
\newtheorem{remark}{Remark}

\theoremstyle{plain}
\newtheorem{proposition}{Proposition}

\newtheorem{theorem}{Theorem}
\newtheorem{lemma}{Lemma}

\begin{document}
\addtolength{\textheight}{0.8cm} 
\addtolength{\voffset}{-0.7cm} 
\begin{frontmatter}
\vspace{25pt}
\title{Handling Pedestrian Uncertainty in Coordinating Autonomous Vehicles at Signal-Free Intersections\thanksref{footnoteinfo}} 
\thanks[footnoteinfo]{This research was supported in part by NSF under Grants CNS-2401007, CMMI-2348381, IIS-2415478, and in part by MathWorks.}
\vspace{0pt}
\author[Paestum1]{Filippos N. Tzortzoglou}\ead{ft253@cornell.edu},    
\author[Paestum1]{Andreas A. Malikopoulos}\ead{amaliko@cornell.edu}
\address[Paestum1]{Department of Civil and Enviromental Engineering, Cornell University, Ithaca, NY 14853, USA}
\begin{abstract}                          
In this paper, we provide a theoretical framework for the coordination of connected and automated vehicles (CAVs) at signal-free intersections, accounting for the unexpected presence of pedestrians. First, we introduce a general vehicle-to-infrastructure communication protocol and a low-level controller that determines the optimal unconstrained trajectories for CAVs, in terms of fuel efficiency and travel time, to cross the intersection without considering pedestrians. If such an unconstrained trajectory is unattainable, we introduce sufficient certificates for each CAV to cross the intersection while respecting the associated constraints. Next, we consider the case where an unexpected pedestrian enters the road. When the CAV's sensors detect a pedestrian, an emergency mode is activated, which imposes certificates related to an unsafe set in the pedestrian's proximity area. Simultaneously, a re-planning mechanism is implemented for all CAVs to accommodate the trajectories of vehicles operating in emergency mode. Finally, we validate the efficacy of our approach through simulations conducted in MATLAB and RoadRunner softwares, which facilitate the integration of sensor tools and the realization of real-time implementation.
\end{abstract}
\end{frontmatter}

\section{Introduction}
\vspace{-5pt}


\small
Emerging mobility systems, such as connected and automated vehicles (CAVs), on-demand mobility services, etc., offer substantial benefits for a sustainable transportation network; see \cite{kamargianni2016critical}. They can improve the quality of life and ensure high safety and efficiency levels; see \citet{zhao2020enhanced}. 
From the perspective of CAVs, the industry has made significant improvements in the last decade. Technologies like lane-keeping systems, emergency braking, and adaptive cruise control have already been employed to enhance safety; see \citet{eichelberger2016toyota}. Similarly, several studies have provided approaches that leverage such technologies to tackle congestion; see \cite{he2019adaptive}. While the results are promising, related research on the coordination of CAVs in traffic scenarios, as intersections and merging on-ramps, often neglects the uncertainty introduced by pedestrians; see \cite{wu2022intersection}. 
Especially, to the best of our knowledge there is not any approach in the literature that considers pedestrian's unexpected crossing at signal-free intersections. Next, we review the relevant literature in this area.
\subsection{Signal-free intersections and CAVs}
A research effort presented by \cite{tachet2016revisiting} showed that autonomous signal-free intersections, namely operating exclusively with CAVs, have the capability of doubling capacity and maximizing throughput in comparison to signalized ones. After the discussion of \cite{tachet2016revisiting}, several research efforts have addressed the problem of coordinating CAVs at signal-free intersection using different methodologies. Some approaches by  \cite{Malikopoulos2017,Malikopoulos2020} and \cite{zhang2019decentralized} proposed coordinating CAVs at signal-free intersections using optimal control. These studies tried to identify optimal trajectories for the CAVs with respect to fuel consumption while minimizing the exit time of each CAV from the intersection with a view to maximizing throughput. However, when using these approaches in congested environments, identifying unconstrained trajectories that do not activate any state or control constraints can be a challenging task for real-time implementation. This applies because the computationally expensive process of piecing together constrained and unconstrained arcs is required; see \cite{Malikopoulos2017} \textit{Section 3.1}. One method to address this problem is by utilizing control barrier functions (CBFs); see \cite{xu2022general,xiao2021bridging} that allow CAVs to avoid collision with other CAVs by sacrificing optimality in the final trajectory. Other research efforts have also utilized reinforcement learning \citep{zhou2019development,farazi2020deep} for the real-time trajectory identification of the CAVs. Along these lines recently a study proposed the utilization of numerical mathematics \cite{Tzortzoglou2024} to address the same problem. 

Finally, some studies have also considered mixed traffic environments where CAVs co-exist with human-driven vehicles; see \citep{gong2018cooperative,keskin2020altruistic,le2024stochastic}, and a combination of routing optimization with coordination of CAVs at traffic scenarios; see \citep{Bang2022combined}.

\subsection{Pedestrian avoidance}
Recent studies have also explored the impact of pedestrians on traffic networks. One early effort exploring pedestrians' behavior was proposed by \cite{hamed2001analysis}, who attempted to analyze the interaction between pedestrians and drivers at pedestrian crossings. Recently, \cite{niels2020simulation,niels2024optimization} proposed a framework where signal phases are defined for pedestrians and subsequently incorporated into a higher-level optimization problem. It is noteworthy that these studies assume pedestrians behave naturally at intersections. Conversely, some studies consider vehicle-pedestrian interaction when the presence of the pedestrian is not anticipated. That is, the pedestrian deviates from normal behavior, interfering with traffic and violating traffic rules.  For example, \cite{coelingh2010collision}  introduced a system called \textit{collision warning with full auto brake and pedestrian detection}, which assists drivers in avoiding both rear-end and pedestrian accidents by providing warnings and, if necessary, automatically applying full braking power. In this system, pedestrian detection takes place through vehicle sensors. Similarly, \cite{schratter2019pedestrian} presented a study where pedestrians are unanticipated and can only be detected through sensors, as well. The vehicle then generates a reachability set that needs to be avoided for each pedestrian and calculates collision-free trajectories. Some other research efforts predict pedestrian intentions using 2D or 3D convolutional neural networks; see \citet{razali2021pedestrian, singh2021multi}.

\subsection{Research gap and contributions of this paper}
While there is a considerable amount of work studying the operation of CAVs at signal-free intersections and pedestrian avoidance separately, to the best of our knowledge, no study has simultaneously addressed the presence of unexpected pedestrians and their impact at signal-free intersections. However, note that according to \cite{nhtsa2024pedestrians}, 16\% of pedestrian fatalities in 2022 occurred at intersection regions (p.~1, bullet 9). Thus, given the continuous integration of CAVs into the market, we believe that the problem of pedestrian safety under uncertainty at signal-free intersections operated by CAVs will attract considerable attention in the coming years.  This is not a trivial problem, as the nature of CAVs implies that an unexpected event may require simultaneous trajectory adjustments from all CAVs at the intersection, along with the implementation of an efficient control policy to avoid the pedestrian without colliding with other CAVs.  In this paper, we aim to address this problem.  The paper is divided into two parts. The first part (\textit{Sections 2-3}) focuses on defining optimal trajectories for the CAVs to safely cross the intersection while minimizing fuel consumption and travel time. The second part (\textit{Sections 4}) introduces a robust control policy for pedestrian avoidance, followed by a re-planning mechanism for all CAVs in the intersection. 

More specifically, in \textit{Section 2} we initially propose the vehicle's dynamics and constraints along with a general vehicle-to-infrastructure communication protocol for any traffic scenario. Next, in \textit{Section 3} we review the theoretical framework presented by \cite{Malikopoulos2020} to define optimal unconstrained trajectories for CAVs to cross a signal-free intersection. However, when an unconstrained trajectory is unattainable, we need to select an efficient method that can determine an alternative trajectory in real time. For that scenario, we select to build upon the work proposed by \cite{xiao2021bridging} and \cite{sabouni2024optimal}, who utilized CBFs to create trajectories that jointly minimize travel time and acceleration while respecting the corresponding constraints with the disadvantage of sacrificing optimality. However, \cite{sabouni2024optimal} assumes a zero standstill distance between vehicles to appropriately define the CBF conditions for vehicles in conflicting paths (see (13) therein and the parameters selection in Section 4). In this paper, we also provide sufficient theoretical arguments for deriving certificates that relax this assumption.

Then, in \textit{Section 4}, we continue our exposition by considering the case of an unexpected pedestrian's presence. When a pedestrian is detected along the trajectory of a CAV, the CAV enters an emergency mode focusing on avoiding the pedestrian. We define the unsafe area associated with the pedestrian, as a function of 1) the vehicle's speed, 2) the distance between the vehicle and the pedestrian, 3) the pedestrian's speed, and 4) the pedestrian's orientation. Based on this unsafe area, we derive certificates that guarantee our vehicle does not enter the pedestrian's unsafe set. These certificates result in CBF conditions (or barrier conditions) that constitute the constraints of a quadratic programming (QP) problem that minimizes the deviation of the actual control input from a reference control input. Note that during the critical event of pedestrian avoidance, some vehicles may deviate from their initial reference trajectories. Therefore, we implement a re-planning mechanism to optimally derive the reference trajectories of all vehicles, in real-time. Finally, we verify our approach through numerical simulations via MATLAB and RoadRunner. 

The contributions of this paper are as follows: 

\begin{itemize}
    \item Expanding the approach proposed by \cite{sabouni2024optimal} to allow CAVs to maintain a nonzero standstill distance when arriving at the intersection from intersecting paths.
    \item Introducing a controller that accounts for pedestrian uncertainty at signal-free intersections operated by CAVs.
    \item Defining an efficient metric for the safe set of each vehicle based on the pedestrian's state, considering the bicycle kinematic model to account for the vehicle's heading angle.
    \item Incorporating in our contoller an algorithm that enables a resequencing--replanning mechanism for the safe operation of CAVs when some deviate from their reference trajectories during the critical event of pedestrian presence.
\end{itemize}

\subsection{Organization of this paper}
The remainder of this paper is organized as follows. Section~2 presents the modeling framework of the CAVs. In Section~3, we introduce our controller in the absence of pedestrians, as previously discussed. Section~4 extends the controller to account for pedestrian detection and defines the resequencing--replanning mechanism. Section~5 provides simulation results using \textsc{Matlab}, and in Section~6 we draw some concluding remarks while providing directions for future work. Additional videos and discussion to facilitate understanding of our approach are available at: \href{https://sites.google.com/cornell.edu/pedestrianavoidance}{https://sites.google.com/cornell.edu/pedestrianavoidance}.
\section{Modeling framework}
Our analysis applies to any traffic scenario, i.e., merging on-ramp, roundabout, intersection etc. We intentionally select an intersection as a reference since it constitutes the most complex case.  The intersection includes a \textit{control zone} (see the left frame in Fig. \ref{fig:Intersection}) where CAVs can exchange information with a coordinator. The coordinator acts only as a database and is not involved in any control of the CAVs. The coordinator may consist of loop detectors or similar sensory devices (road units) capable of monitoring the state of CAVs within the control zone. We thoroughly discuss the communication protocol between CAVs and the coordinator in Section 2.2. 

Let $\mathcal{N}(t)=\{1,2,\dots,\textit{N}\},\;\textit{N}\in \mathbb{N},$ be the set of CAVs in the control zone at time $t$. Let $\mathcal{G}=\{1,2,\ldots,z,\ldots ,|\mathcal{G}| \},\; z \in \mathbb{N}$, be the set of the finite paths in the control zone crossing the intersection (see, for example, the right frame in Fig. \ref{fig:Intersection}). In our approach, we consider that CAVs follow their desired path before entering the control zone, preventing lane-change maneuvers. Although this consideration is reasonable given the nature of CAVs, it can be easily relaxed; see \cite{Malikopoulos2020}.    Lateral collisions may occur at the locations where different paths intersect, which we refer to as \textit{conflict points}. These points are denoted by the set $\mathcal{O}_z \subseteq \mathbb{N}$ for each path $z \in \mathcal{G}$. Namely, $\mathcal{O}_z$ contains the set of conflict points on path~$z$. For instance, in the upper right frame in Fig.~\ref{fig:Intersection}, path $z_1$ will pass along 4 conflict points where the last conflict point is $n_1$. Note that conflict point $n_1$ is the first conflict point for the path $z_2$, i.e., $n_1$ will be the first element of the set $\mathcal{O}_{z_2}.$ It is important to note that these points remain fixed for any state within the control zone. Finally, let $L_z^n$ be the distance of conflict point $n$ from the entry of the control zone for path $z$. 

\begin{figure*}
    \centering
    \begin{subfigure}[b]{0.622\textwidth}
        \includegraphics[width=\textwidth]{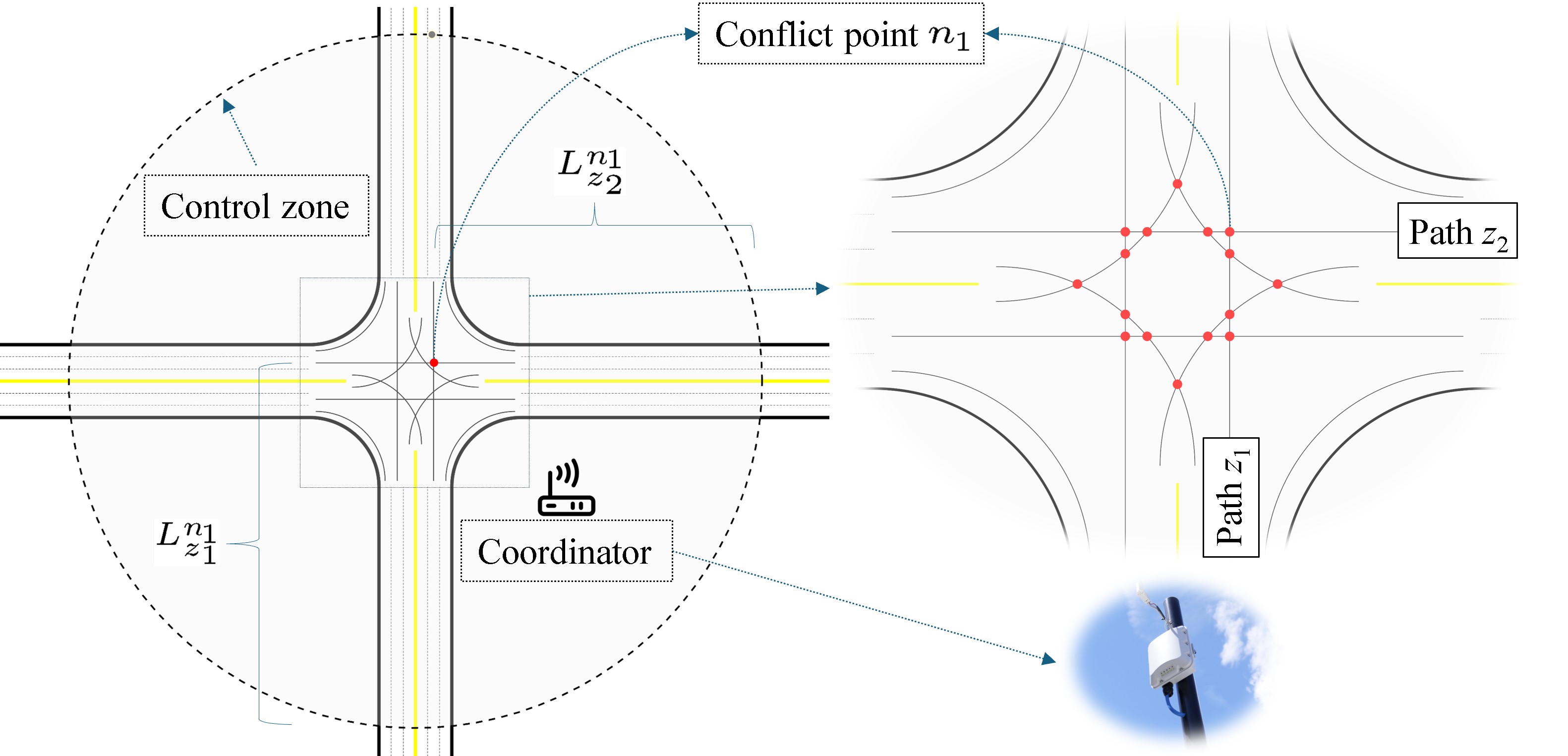}
        \caption{Illustration of a signal-free intersection along with the control zone, conflict points, and coordinator.}
        \label{fig:Intersection}
    \end{subfigure}%
    \hspace{4mm} 
    \parbox[b]{0.35\textwidth}{
        \centering
        \begin{subfigure}[b]{0.26\textwidth}
            \includegraphics[width=\textwidth]{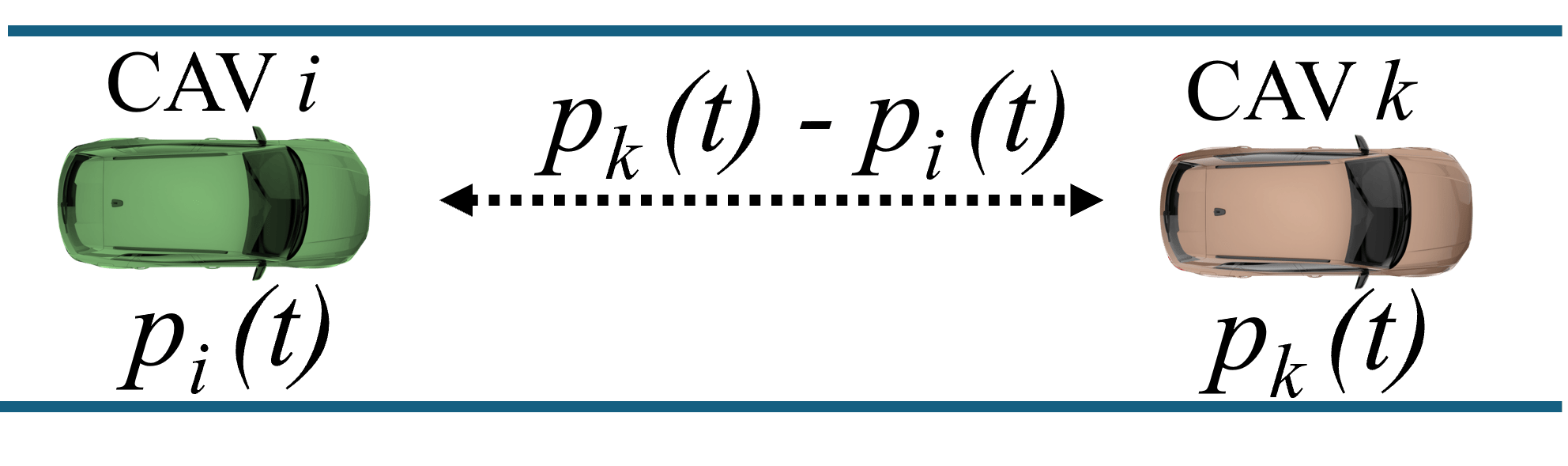}
            \caption{Illustration of rear-end constraint.}
            \label{fig:rear_end1}
        \end{subfigure}
        \vspace{2mm} 
        \begin{subfigure}[b]{0.26\textwidth}
            \includegraphics[width=\textwidth]{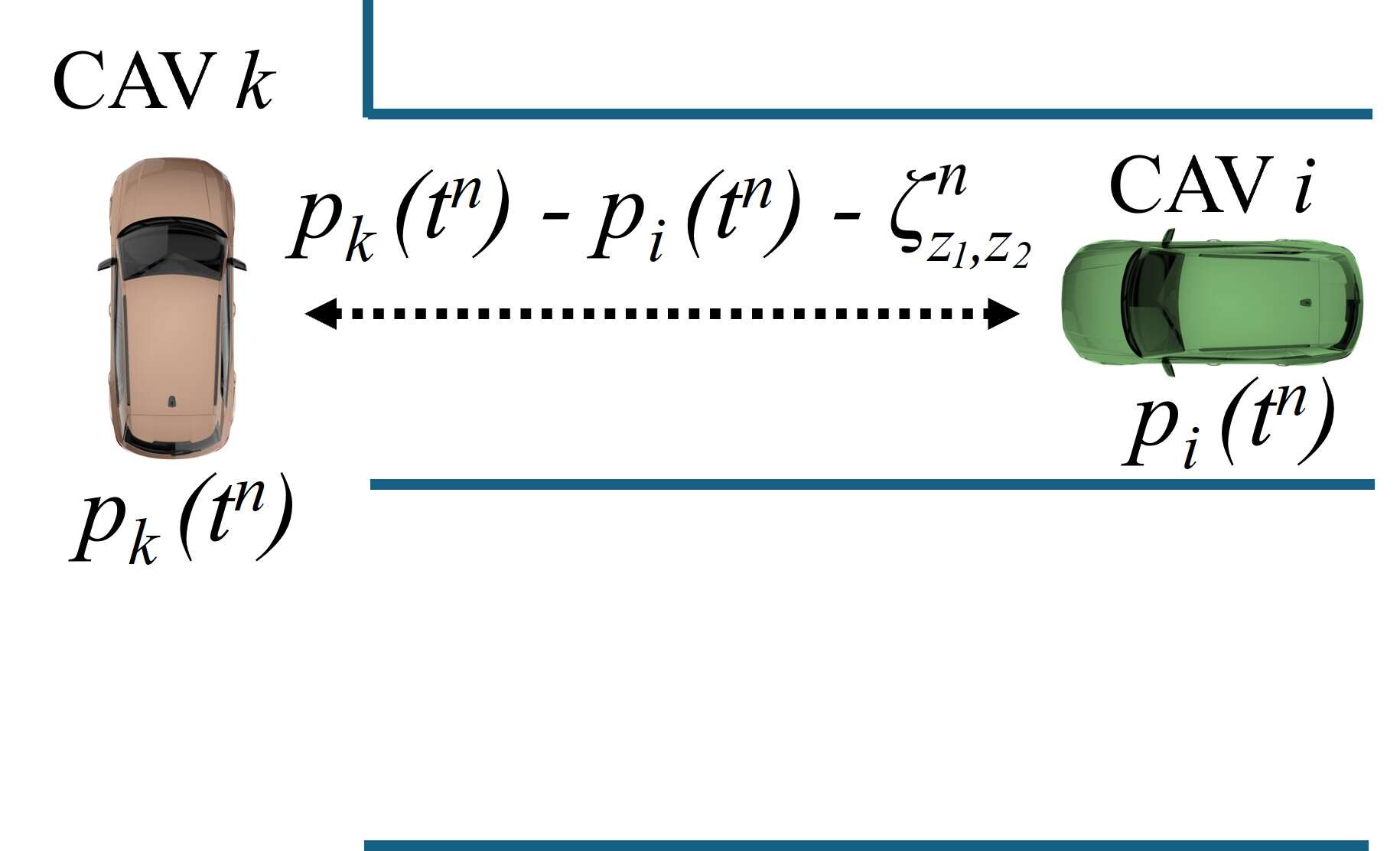}
            \caption{Illustration of lateral constraint.}
            \label{fig:lateral}
        \end{subfigure}
    }
    \caption{Set up of intersection and illustration of rear-end and lateral constraints.}
\end{figure*}

Before we present our dynamics, we define two assumptions that we consider throughout our exposition.

\begin{itemize}
    \item Assumption 1: Tracking controllers are available at the low-level control layer of CAVs to ensure reference speed tracking and steering for lane keeping. 
    \item Assumption 2: Communication noise and delay between CAVs and traffic lights is negligible.
\end{itemize}

We impose Assumption 1 to focus explicitly on higher-level decision-making, though deviations could easily be incorporated if needed; see \cite{ChalakiCBF2022}. Assumption 2 can be similarly relaxed by introducing communication noise, as demonstrated in previous work \cite{chalaki2021RobustGP}.

\subsection{Vehicle Dynamics}
\vspace{-5pt}
We model the dynamics of each CAV as a double-integrator 
\begin{equation}\label{eq:model2}
\begin{split}
\dot{p}_{i}(t) &= v_{i}(t), 
\\
\dot{v}_{i}(t) &= u_{i}(t), 
\end{split} 
\end{equation}
where ${p_{i}\in\mathcal{P}}$, ${v_{i}\in\mathcal{V}}$, and
${u_{i}\in\mathcal{U}}$ denote the position of the rear bumper from the entry of the control zone, speed, and control input (acceleration/deceleration) of the CAV $i$, respectively. Note that we omit the independent variable $t$ from our notation when it does not lead to ambiguity. 
\noindent Next, let $t_i^0$ be the time a CAV~$i$ enters the control zone and $t_i^f$ the time it exits the control zone. Also, let $t_i^n$ be the time CAV $i$ reaches conflict point $n \in \mathcal{O}_z $ and $p_i^n$ be the position of CAV $i\in \mathcal{N}$ at conflict point $n$. Let $p_i^0$ and $p_i^f$ be the position of the vehicle at the entry and the end of the control zone, respectively. Without loss of generality we consider $p_i^0=0$.

In the following subsection, we introduce all the constraints that can influence the safe operation of the CAVs. 

\subsubsection{Safety constraints}
To prevent any possible rear-end collision between two consecutive vehicles $i\in\mathcal{N}$ and $k\in \mathcal{N}\backslash \{i\}$ where $k$ is the preceding one, we consider the following constraint:
\begin{equation}\label{eq:rear-end}
    p_k(t)-p_i(t)\geq \delta_i(t) = \phi v_i(t) + \gamma,
\end{equation}
\noindent where $\delta_i(t)$ is the safe-dependent distance, while $\phi$ and $\gamma$ are the reaction time and standstill distance, respectively. To visualize the constraint see Fig \ref{fig:rear_end1}.

We also impose lateral constraints on CAVs traveling on intersecting paths. For instance, in Fig. \ref{fig:Intersection}, paths $z_1$ and $z_2$ intersect at conflict point $n_1$, requiring the corresponding CAVs to avoid a collision at this point. To rigorously define this constraint, we consider a CAV $k\in \mathcal{N}$ on path $z_i\in\mathcal{G}$ that might cause a collision with CAV $i\in \mathcal{N}$ on path $z_j\in \mathcal{G}$ at conflict point $n$. Then, at time $t^n = \text{min}(t_k^n, t_i^n),$ (\( t^n \) represents the time at which CAV \( i \) or CAV \( k \) is planned to first cross conflict point \( n \)) we obtain the lateral constraint: 
\begin{equation}
| p_k(t^n) - p_i(t^n) | - \zeta_{z_\iota, z_j}^n \geq \phi v_i(t^n) + \gamma, \; \; \;
\end{equation}
which guarantees that once a CAV reaches a conflict point $n$, a potentially conflicting CAV will have a distance from the conflict point greater than or equal to $\phi v_i(t^n) + \gamma$. Note that the distance of a conflict point from the entries of two conflicting paths in the control zone might not be the same. Since the function \( p_i(t) \) tracks the distance from the control zone entry, this difference must be considered. Thus, we introduce the variable $\zeta_{z_\iota, z_j}^n,\;z_\iota,z_j \in \mathcal{G}$ that denotes the relative distance difference of a conflict point $n$ from the entries of two conflicting paths $z_\iota$ and $z_j$, in the control zone. For example in Fig. \ref{fig:Intersection}, for the conflicting paths $z_1$ and $z_2$ with conflict point $n_1$, the value of $\zeta_{z_1, z_2}^{n_1}$ is equal to $|L_{z_1}^{n_1}-L_{z_2}^{n_1}|$, where $L_{z_1}^{n_1} \neq L_{z_2}^{n_1}$. For an additional visualization see Fig \ref{fig:lateral} where a CAV $k$ crosses a conflict point before CAV $i$. Also, note that for each combination of intersecting paths we have only one conflict point in common and we can just use $\zeta_{z_\iota, z_i}$. Given that $\zeta_{z_\iota, z_j}$ is a constant we can also incorporate $\zeta_{z_\iota, z_j}$ into $\gamma$  and obtain:
\begin{equation}\label{eq:lateral}
    | p_k(t^n) - p_i(t^n) | \geq \phi v_i(t^n) + \tilde{\gamma}, \; \; \; t^n = \text{min}(t_k^n, t_i^n),
\end{equation}
where $\tilde{\gamma}=\gamma+\zeta_{z_\iota, z_j}$. Observe that although $\tilde{\gamma}$ is a function of $z_i$ and $z_j$, we omit writing it as $\tilde{\gamma}(z_\iota, z_j)$ (with a slight abuse of notation) for brevity, as this does not lead to any ambiguity. Finally, note that the absolute value is used in \eqref{eq:lateral} because the crossing sequence of CAVs $i$ and $k$ is not predetermined. This applies thanks to the absence of a strict sequence rule.


We also consider the speed limits and the physical
acceleration and braking limits of the vehicles. Consequently, for each CAV $i\in\mathcal{N}$ we impose
\begingroup
\setlength{\belowdisplayskip}{10pt} 
\setlength{\abovedisplayskip}{10pt} 
\begin{align}\label{eq:speed}
    v_{\text{min}}\leq v_i(t)\leq v_{\text{max}}, \quad
     u_{i,\text{min}} \leq u_i(t) \leq u_{i,\text{max}}. 
\end{align}
Next, we define the communication protocol inside the control zone and review the control framework presented by \cite{Malikopoulos2020}.

\subsection{Communication protocol within the control zone}
Once CAV $i$ enters the control zone, information can be exchanged with the coordinator and other CAVs. Initially, the CAV shares the path $z\in \mathcal{G}$ on which it operates with the coordinator, and at the same time, it receives the tuple $\mathcal{I}_i(t)=\{\mathcal{C}_z,S\}$. Here, $\mathcal{C}_z$ includes the time intervals when path $z$ makes any of the rear-end or lateral constraints active, and $S$ constitutes the speed limit of the intersection at time $t$. After the CAV has defined its trajectory based on $\mathcal{I}_i(t)$, it shares the trajectory with the coordinator. In case where two CAVs enter the control zone simultaneously, we consider that the order of trajectory planning happens with equal probability. Finally, note that the communication with the coordinator can be executed more than once while a CAV is in the control zone. Especially, in the case of resequencing-replanning; see \citet{chalaki2021Reseq}, the vehicles need to redefine their reference trajectories while they are inside the control zone, thus requiring access to the coordinator more than once. A resequencing-replanning mechanism is discussed in detail in Section 4.4.   

\section{Low-level controller design for CAVs}
In this section, we review the low-level controller introduced by \cite{Malikopoulos2020}, which is used in our framework to derive optimal unconstrained trajectories for the CAVs crossing the intersection such that none of the state, control, and safety constraints are active.
\subsection{Definition of low-level controller}

The goal of each CAV is to minimize its travel time while minimizing its control input (acceleration). These objectives aim to maximize throughput at the intersection while indirectly minimizing fuel consumption. To achieve this, we formulate two optimal control problems: (1) for each CAV $i$, an upper-level optimal control problem identifies the exit time from the control zone (\textit{time-optimal control problem}), which is then passed as an input to (2) a low-level optimal control problem that yields the control input trajectory that minimizes fuel consumption (\textit{energy-optimal control problem}). The low-level optimization problem is formulated as follows:

\textbf{Low-level (energy-optimal) control problem:} \label{prb:ocp-1}
\begin{equation}
\begin{aligned}
\label{eq:energy_cost}
&\underset{u_i\in\mathcal{U}}{\min} \quad \frac{1}{2} \int_{t^{0}_{i}}^{t_i^f} u^2_i(t) \, \mathrm{d}t, \\
&\text{subject to:} \quad \eqref{eq:model2}, \eqref{eq:rear-end}, \eqref{eq:lateral}, \eqref{eq:speed},\\
&\text{given:} \quad p_i (t_i^0) = p_i^0, \,\, v_i (t_i^0) = v_i^0, 
\,\, p_i (t_i^f) = p_i^f,
\end{aligned}
\end{equation}
where $t_i^f$ is computed by the upper-level optimization problem, discussed next. This problem can be solved analytically by utilizing the \textit{Pontryagin's minimum principle} as presented by \cite{Malikopoulos2020}. The optimal unconstrained trajectories are 
\begin{equation}\label{eq:optimalTrajectory}
\begin{split}
u_i(t) &= 6 \phi_{i,3} t + 2 \phi_{i,2}, \\
v_i(t) &= 3 \phi_{i,3} t^2 + 2 \phi_{i,2} t + \phi_{i,1}, 
\\
p_i(t) &= \phi_{i,3} t^3 + \phi_{i,2} t^2 + \phi_{i,1} t + \phi_{i,0},
\end{split}
\end{equation} 
where $\phi_{i,3}, \phi_{i,2}, \phi_{i,1}, \phi_{i,0}$ are constants of integration and can be found using the initial and terminal conditions along the optimal terminal condition $u_i(t_i^f)=0$. \\ \\
\textbf{Upper-level (time-optimal) control problem:}
At time \( t_i^0 \), when CAV enters the control zone, let $ \mathcal{F}_i(t_i^0) = [\underline{t}_i^f, \overline{t}_i^f] $ be the feasible range of travel time under the speed and input constraints of CAV \( i \) computed at \( t_i^0 \). The values \( \underline{t}_i^f \) and \( \overline{t}_i^f \) correspond to the time that the CAV exits the control zone with the maximum and minimum acceleration subject to \eqref{eq:optimalTrajectory}, respectively. The complete analysis of calculating $ \mathcal{F}_i(t_i^0) = [\underline{t}_i^f, \overline{t}_i^f] $ was defined by \cite{chalaki2021CSM} (p. 27).  Then, CAV \( i \) must solve the following time-optimal control problem to find the minimum exit time \( t_i^f \in \mathcal{F}_i(t_i^0) \) that satisfies all the constraints,
\begin{align}\label{Time_optimal}
&\underset{t_i^f \in \mathcal{F}_i(t_i^0)}{\min} \quad t_i^f  \\
&\text{subject to:} \quad \eqref{eq:model2}, \eqref{eq:rear-end}, \eqref{eq:lateral}, \eqref{eq:speed}, \eqref{eq:optimalTrajectory} \nonumber \\
&\text{given:} \quad p_i(t_i^0) = p_i^0, \, v_i(t_i^0) = v_i^0, \, p_i(t_i^f) = p_i^f, \, u_i^*(t_i^f) = 0. \nonumber 
\end{align}
\noindent The boundary conditions in \eqref{Time_optimal} are set at the entry and exit of the control zone, and $u_i^*(t_i^f) = 0$ denotes the optimal acceleration at $t_i^f$. 

The two optimal control problems aim to find an unconstrained trajectory using the following procedure. Initially, the time $t_i^f=\underline{t}_i^f$ is set, and then \eqref{eq:energy_cost} is solved. If none of the constraints \eqref{eq:rear-end}, \eqref{eq:lateral}, \eqref{eq:speed} is violated, the solution is given by \eqref{eq:optimalTrajectory}. Otherwise, we set $t_i^f = t_i^f + \Delta t$, where $\Delta t$ is a timestep, and we check again if a constraint is violated. We continue this process until no constraint is violated.  For a visual demonstration of how this process works, the reader can refer to the paper's website. 
\begin{remark}
The process described above allows us to determine an optimal unconstrained trajectory for each CAV $i$. However, in congested environments, there may be cases where an unconstrained trajectory is not feasible within $ \mathcal{F}_i(t_i^0) = [\underline{t}_i^f, \overline{t}_i^f] $ and the need of piecing together arcs is inevitable if we want to guarantee an optimal solution. An approach to circumvent the computational complexities associated with this process was given by \cite{xiao2021bridging} and \cite{sabouni2024optimal}, who attempted to define trajectories that satisfy all the constraints using CBFs, with the disadvantage of sacrificing optimality. In this paper we build upon their work, and we also extend their approach, which assumed a standstill distance of zero between consecutive vehicles to efficiently define lateral constraints. In the next subsection, we will relax this constraint.
\end{remark}
\subsection{Certificates for constraint satisfaction.}
In this subsection, we associate the constraints presented in our problem formulation with barrier conditions (or equivalently CBF conditions) that are linear in the control input.  The key idea is that these barrier conditions ensure the forward invariance property of the constraint set. That is, if a constraint set $C$ defines the feasible region for the system states, the corresponding barrier conditions enforce that, for any initial condition within $C$, the system remains within $C$ for all future time steps. Let us recall the definition of the forward invariance property.

\begin{definition}
\textbf{(\cite{khalil1992nonlinear})}. A set $C \subset \mathbb{R}^n$ is forward invariant for a system if the solutions of the trajectories starting at any $\mathbf{x}(0) \in C$ satisfy $\mathbf{x}(t) \in C, \forall t \geq 0$.
\end{definition}

Thus, our goal is to define barrier conditions that satisfy the invariance property and then solve an optimization problem that minimizes the deviation of the control input from a reference (optimal) control input, subject to these barrier conditions. Note that since these barrier conditions are linear in the control input, we obtain a QP problem that can be executed in real time. We now define these barrier conditions.

To maintain consistency with the literature in this area, we generalize our system dynamics in the vector form:
\begin{align}
\textbf{x}_i(t) = f(\textbf{x}_i(t)) + g(\textbf{x}_i(t))\textbf{u}_i  
\end{align}
where, $\textbf{x}_i(t)=[p_i(t),v_i(t)]^T$, $f(\textbf{x}_i(t))=[v_i(t),0]^T$ and $g(\textbf{x}_i(t))=[0,1]^T$. Note that $f$ denotes the drift vector field, and $g$ denotes the control vector field. Next, we need to transform each of our constraints \eqref{eq:rear-end}, \eqref{eq:lateral} and \eqref{eq:speed} into the CBF form $b_q(\textbf{x}(t))\geq0$ where $q \in \{1,...,c\}$, and $c\in\mathbb{N}$ stands for the number of constraints. As it was shown by \cite{ames2016control,ames2019control}, after transforming each of our constraints in the form $b_q(\textbf{x}(t))\geq0$, we can define a new equivalent constraint that we call \textit{barrier condition} which is linear in the control input and satisfies the forward invariance property. This barrier condition is defined as follows: 
\begin{equation}    \mathcal{L}_fb_q(\textbf{x}(t))+\mathcal{L}_gb_q(\textbf{x}(t))u(t)+\alpha(b_q(\textbf{x}(t)))\geq 0,
    \label{CBF constraint}
\end{equation}
where  the symbol $\mathcal{L}$ stands for the Lie derivatives. The inequality \eqref{CBF constraint} constitutes the certificate associated with the constraint $b_q(\textbf{x}(t))\geq0$. 
The Lie derivative of a function $b(\textbf{x})$ with respect to the vector field $f$ is given by $\mathcal{L}_fb(\textbf{x})=\langle \nabla b(\textbf{x}), f(\textbf{x})  \rangle$ where $\nabla b(\textbf{x})$ is the gradient of the function $b(\textbf{x})$, i.e., the vector of partial derivatives of $b(\textbf{x})$ with respect to the state variables. Finally, $\alpha$ stands for the class $\mathcal{K}$ function which is chosen here as the identity function. Let us recall the definition of a class $\mathcal{K}$ function.

\begin{definition}
\textbf{(\cite{khalil1992nonlinear})}. A continuous function $\alpha: [0, \eta) \to [0, \infty)$, $\eta > 0$ is said to belong to class $\mathcal{K}$ if it is strictly increasing and $\alpha(0) = 0$.
\end{definition}

Next, we transform each of our constraints in the form $b_q(\textbf{x}(t))\geq0$ and define the associated barrier conditions based on \eqref{CBF constraint}. We begin with the rear-end constraint. Recall that the rear-end safety constraint is $p_k(t)-p_i(t)\geq \phi v_i(t) + \gamma$ that can be written as $b_1(\textbf{x}_i(t))=p_k(t)-p_i(t)-\gamma -\phi v_i(t) \geq 0$. Then, according to \eqref{CBF constraint} we obtain
\begin{equation}
\underbrace{v_k(t) - v_i(t)}_{\mathcal{L}_{f}b_1} + \underbrace{ - \phi}_{\mathcal{L}_{g}b_1} u_i(t) + \underbrace{p_k(t) - p_i(t) - \gamma - \phi v_i(t)}_{\alpha (b_1)=b_1} \geq 0.
    \label{CBFrearend1}
\end{equation}
Here $\alpha (b_1)=b_1$ confirms that the class $\mathcal{K}$ function is the identity function, as previously discussed. Similarly, we obtain sufficient barrier conditions for the speed limits of the vehicles. Note that the limits related to the speed of the vehicles can be written as $b_2(\textbf{x}_i(t))=v_{\text{max}}-v_i(t)\geq 0$ and $b_3(\textbf{x}_i(t))=v_i(t)-v_{\text{min}} \geq 0$. Thus, the barrier conditions associated with the speed limits are
\begin{align} 
    u_i(t) &\leq v_{\text{max}}-v_i(t), \label{Cert_acc_1}\\
     u_i(t) &\geq -v_i(t)+v_{\text{min}} \label{Cert_acc_2}. 
\end{align}
For the lateral constraints, there is a technical difficulty in deriving sufficient barrier conditions since these constraints apply only at conflict points and not continuously. To make that clearer recall that in \eqref{eq:lateral} we only evaluate the constraint at the moment $t^n$  where the first of two conflicting CAVs arrives at a conflicting point. However, CBFs require a continuous differentiable form for each constraint; see \cite{ames2016control}. Therefore, we must effectively transform constraint \eqref{eq:lateral} into a continuous differentiable format. \cite{sabouni2024optimal} suggested using a reaction-time-based approach, by creating a function $\Phi(p_i(t))$ that smoothly outputs the reaction time $\phi$ based on the position $p_i(t)$. This function operates over a time interval $[t_i^0, t_i^n]$, during which CAV $i$ travels from the entry of the control zone until the corresponding conflict point $n$ (or merging point as defined by \cite{sabouni2024optimal}). They set initial and terminal conditions as $\Phi(p_i(t_i^0))=0$ and $\Phi(p_i(t_i^n))=\phi$. Using these conditions, the lateral constraint was expressed in a continuous form as 
\begin{equation} \label{reaction_based_lateral}
  p_k(t) - p_i(t) \geq \Phi(p_i(t))v_i(t) + \tilde{\gamma} \;\forall t \in [t_i^0, t_i^n],  
\end{equation} which implies that at $t = t_i^0$, \eqref{reaction_based_lateral} yields $p_k(t) - p_k(t) \geq0 \; ,\text{if} \;\tilde{\gamma} = 0$. And at $t = t_i^n$, \eqref{reaction_based_lateral} yields $p_k(t) - p_i(t) \geq \phi v_i(t) + \tilde{\gamma}$. Thus, at the time $t_i^n$, the function $\Phi(p_i(t))$ results in the same constraint as before the transformation. \cite{sabouni2024optimal} used the function $\Phi(p_i(t)) = \frac{\phi p_i(t)}{L}$ to meet these requirements, where $L$ was the distance from the entry point of the control zone to the conflict point (in our case $L_z^n$). However, note that if $\tilde{\gamma} \neq 0$, this function does not yield the desired result at $t = t_i^0$. Especially at $t = t_i^0$, for $\tilde{\gamma}>0$, the lateral constraint becomes $p_k(t) - p_i(t) \geq \tilde{\gamma} > 0 $ , meaning that when CAV $k$ arrives at the control zone, CAV $i$ should be positioned $\tilde{\gamma}$ meters away from the entrance point and vice versa. This assumption imposes strict timing on vehicle arrivals.
To address this limitation, in the following Lemma we introduce a new format for the function $\Phi(p_i(t))$ that accommodates any value of $\tilde{\gamma}$ and thus relaxes the assumption on entrance vehicle spacing.
\begingroup   
\begin{lemma} \label{extension lemma}
       Consider a CAV $i\in \mathcal{N}$ operating on a path $z\in \mathcal{G}$.  Let $\Phi(p_i(t))$ be a continuously differentiable function satisfying the boundary conditions $\Phi(p_i(t_i^0)) = -\frac{\tilde{\gamma}}{v_i^0}$ and $\Phi(p_i(t_i^n)) = \phi$, where $p_i(t_i^0)=0$ is the position of the vehicle at the initial time $t_i^0$ and $p_i(t_i^n) = L_z^n$ is the position at time $t_i^n$, when CAV $i$ will reach conflict point $n$. Also $\tilde{\gamma}$, $\phi$, and $v_i^0$ are constants. Then 
        \begin{equation}
        \Phi(p_i(t)) = \frac{\phi + \frac{\tilde{\gamma}}{v_i^0}}{L_z^n}p_i(t) - \frac{\tilde{\gamma}}{v_i^0},
        \end{equation}    
         satisfies the boundary conditions, and the right-hand side of \eqref{eq:lateral} at $t=t_0$ is $\Phi(p_i(t_0)) v_i(t_0) + \tilde{\gamma} = 0 ;\forall \tilde{\gamma} \in \mathbb{R}^+.$
\end{lemma}
\endgroup
\begin{proof}
    Given that we seek to identify a function to eliminate the term $\tilde{\gamma}$ at $t=t_i^0$ regardless of its value in \eqref{eq:lateral}, we choose $\Phi$ such as  $\Phi(p_i(t)) = a p_i(t) + b$. From $\Phi(p_i(t_i^0)) = -\frac{\tilde{\gamma}}{v_i^0}$,  we obtain 
    $-\frac{\tilde{\gamma}}{v_i^0} = a p_i(t_i^0) + b$, hence $b = -\frac{\tilde{\gamma}}{v_i^0} - a p_i(t_i^0).$
    From $\Phi(p_i(t_i^n)) = \phi$, we obtain
    $\phi = a L_z^n + b.$
    Substituting $b$ into the last equation equation, we obtain
    $\phi = a L_z^n - \frac{\tilde{\gamma}}{v_i^0} - a p_i(t_i^0)$; hence,
    $
    a = \frac{\phi + \frac{\tilde{\gamma}}{v_i^0}}{L_z^n - p_i(t_i^0)}.
    $
    Substituting $a$ and $b$ into $\Phi(p_i(t))$, given $p_i(t_i^0)=0$ , $p_i(t_i^n) = L_z^n$  and simplifying yields
    $
    \Phi(p_i(t)) = \frac{\phi + \frac{\tilde{\gamma}}{v_i^0}}{L_z^n} p_i(t) - \frac{\tilde{\gamma}}{v_i^0}.
    $ Thus lateral constraint at $t=t_i^0$ becomes 
    $p_k(t)-p_i(t)\geq \frac{\tilde{\gamma}}{v_i^0}v_i^0 - \tilde{\gamma} = 0 \; \forall \tilde{\gamma} \in \mathbb{R}^+$, which completes the proof.
\end{proof}

At this juncture, one may substitute the function \( \Phi(p_i(t)) \) as defined in Lemma \eqref{extension lemma} with \( \phi \) in \eqref{eq:lateral}. This modification permits the articulation of the lateral constraint in a continuously differentiable format as follows:
\begin{align} \label{b_4_safe_set}
     b_4(\textbf{x}(t))& = |p_k(t) - p_i(t)| - \Phi(p_i(t)) v_i(t) -\tilde{\gamma}   \nonumber \\
    & = |p_k(t) - p_i(t)| - ( \frac{\phi + \frac{\tilde{\gamma}}{v_i^0}}{L_z^n} p_i(t) - \frac{\tilde{\gamma}}{v_i^0})v_i(t) -\tilde{\gamma} \geq 0. 
\end{align}

In the following result, we define our berrier condition for the lateral constraint, which imposes no restrictions on the vehicles' standstill distance.
\begin{theorem}
Consider CAV $i\in\mathcal{N}$ and CAV $k\in\mathcal N\backslash \{i\}$ operating at intersecting paths. Consider that CAV $k$ reaches first at conflict point $n$. The barrier condition associated with $\eqref{b_4_safe_set}$ that guarantees lateral collision avoidance between the two vehicles is given by
\begingroup
\begin{align} \label{Cert_lateral}
    & v_k(t)-v_i(t) - \frac{\phi + \frac{\tilde{\gamma}}{v_i^0}}{L_z^n}v_i(t)^2 - ( \frac{\phi + \frac{\tilde{\gamma}}{v_i^0}}{L_z^n} p_i(t) - \frac{\tilde{\gamma}}{v_i^0})u_i(t) \nonumber \\ 
    & + p_k(t) - p_i(t) - ( \frac{\phi + \frac{\tilde{\gamma}}{v_i^0}}{L_z^n} p_i(t) - \frac{\tilde{\gamma}}{v_i^0})v_i(t) -\tilde{\gamma}\geq 0.
\end{align}
In addition, \eqref{Cert_lateral} does not possess any restrictions on the standstill distance and at the arrival times of the vehicles. 
\end{theorem}
\endgroup
\begin{proof}
The Lie derivatives of $b_4(x(t))$ in \eqref{b_4_safe_set} with respect to the vector fields $f$ and $g$ are $ \mathcal{L}_fb_4 = v_k(t)-v_i(t) - \frac{\phi + \frac{\tilde{\gamma}}{v_i^0}}{L_{z_2}^n}v_i(t)^2 $ and $\mathcal{L}_gb_4 = - ( \frac{\phi + \frac{\tilde{\gamma}}{v_i^0}}{L_z^n}p_i(t) - \frac{\tilde{\gamma}}{v_i^0})u_i(t)$, respectively. Since \eqref{b_4_safe_set} is derived by Lemma~1, there are no restrictions to the standstill distance and vehicles' arrival time. By applying \eqref{CBF constraint}, the result follows. 
\end{proof}
\begingroup
Given all the barrier conditions, we obtain the following optimization problem:
\begin{align} \label{QP_problem}
\underset{u_i\in\mathcal{U}}{\min} &  \int_{t_i^0}^{t_i^f} \frac{1}{2}(u_i(t)-u_\textbf{ref}(t))^2 dt \\
    \text{subject to: }  &   \eqref{CBFrearend1},  \eqref{Cert_acc_1},  \eqref{Cert_acc_2}, \text{and} ~\eqref{Cert_lateral}. \nonumber
\end{align}
\endgroup
This optimization problem aims at minimizing the deviation of our control input $u_i(t)$ from a reference control trajectory $u_{i, \text{ref}}(t)$ in the time interval $[t_i^0, t_i^f]$. This reference control trajectory can be any optimal unconstrained trajectory. Note that \cite{xu2022general} and \cite{sabouni2024optimal} utilized the optimal unconstrained trajectory that jointly minimizes travel time and control input.  Defining such an unconstrained trajectory utilizes again the Pontryagin's minimum principle as in \eqref{eq:energy_cost} with the only difference of incorporating the final time $t_i^f$ as a decision variable. Finally, to solve \eqref{QP_problem}, we discretize the domain $[t_i^0, t_i^f]$ into multiple discrete time steps $[t_i^0,t_i^0+\Delta t, t_i^0+2\Delta t, ...,t_i^f] $. This allows us to evaluate the problem at discrete intervals. Since each constraint is linear in the control input, the resulting optimization problem is a QP, which can be efficiently solved in real time.

Up to this point, our analysis has focused solely on the operation of CAVs at signal-free intersections, without accounting for the presence of unexpected pedestrians. Specifically, in \textit{Section 3.2}, we discussed how to identify unconstrained trajectories, and subsequently presented an approach for defining trajectories based on the works of \cite{xu2022general} and \cite{sabouni2024optimal} when unconstrained trajectories are not feasible. Also we relaxed an assumption regarding the standstill distance of the vehicles. Following, we consider pedestrian's unexpected crossing.

\section{Case of unexpected event}
In this section, we consider the case of an unexpected event, such as a pedestrian disregarding traffic rules to cross the road. The high-level overview of the following analysis is as follows: once CAV $i \in \mathcal{N}$ detects the pedestrian, it broadcasts a critical event alert to all nearby CAVs. Those traveling on the same road as the pedestrian immediately activate an emergency mode, while the remaining CAVs at the intersection sequentially replan their reference trajectories to account for the altered paths of the CAVs in emergency mode. In this emergency mode, each CAV follows a trajectory that prioritizes pedestrian safety. 

In the next subsections, we initially discuss in detail the control strategy during emergency mode and derive our re-planning mechanism that allows all the CAVs to safely coordinate during the critical event.

\begin{remark}
When a CAV enters an emergency mode due to the detection of a pedestrian, it is crucial to incorporate its steering angle in our analysis. Thus, under emergency mode, we relax \textit{Assumption 1} regarding steering tracking and treat the heading angle of the CAV as a control input. This provides an extra degree of freedom, which can be critical. We account for the steering angle by considering the bicycle kinematic model.
\end{remark}

\begin{remark}
As shown in the videos available on the paper's website and discussed in the simulation results, we consider that each CAV is equipped with at least a LiDAR and a camera sensor, enabling it to detect and classify pedestrians with high accuracy. This is a reasonable assumption, commonly adopted in many autonomous vehicles to date; see \cite{wang2019pseudo}.
\end{remark}

\subsection{Bicycle kinematic model} The dynamics of the bicycle kinematic model are presented in \eqref{Bicycle model}. Here, $x$ and $y$ represent the longitudinal and lateral positions, measured at the midpoint of the rear axle while $v$ denotes the linear speed, $\theta$ the heading angle, $\delta$ the steering angle (first input), $u$ the acceleration (second input), and $\sigma$ the vehicle length.
\begin{figure}[h!]
    \centering
    \hspace{0.05\linewidth} 
    \begin{minipage}{0.45\linewidth}
        \centering
        \includegraphics[width=1\linewidth]{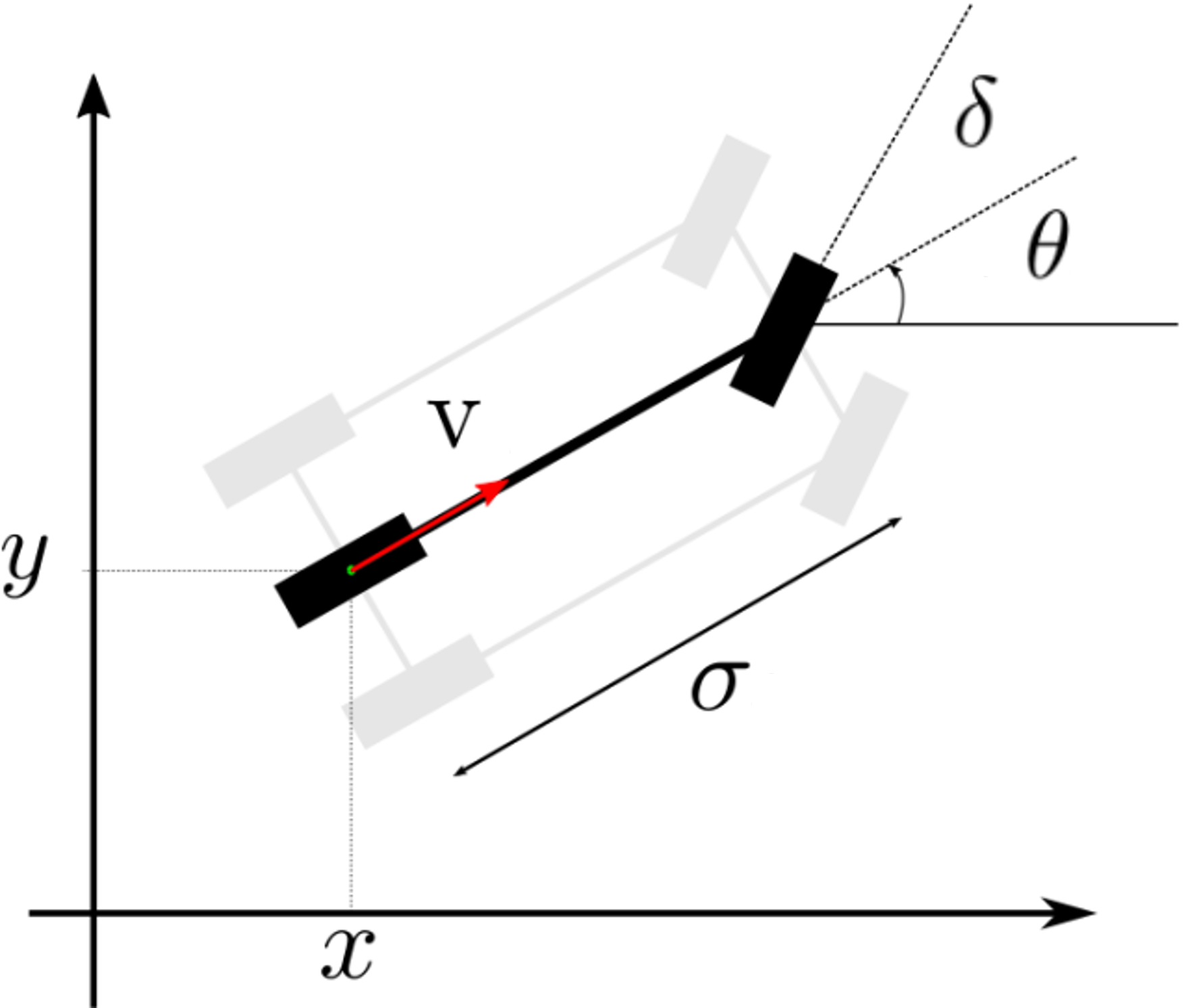}
        \label{bicycle_model}
    \end{minipage}%
    \hfill
    \raisebox{0.3\height}{ 
    \begin{minipage}{0.45\linewidth}
        \centering
        \begin{align}\label{Bicycle model}
        \dot{x} &= v \cos(\theta) \nonumber \\ 
        \dot{y} &= v \sin(\theta) \nonumber  \\ 
        \dot{\theta} &= \frac{v}{\sigma} \tan(\delta) \\ 
        \dot{v} &= u \nonumber
        \end{align}
    \end{minipage}
    }
\caption{Bicycle kinematic model and dynamics for CAV i.}
\end{figure}
However, the model in \eqref{Bicycle model} is not in control-affine form, (namely, it cannot be written as $\textbf{x} = f(\textbf{x}) + g(\textbf{x})\textbf{u}$), since the control input $\delta$ is an argument of the function $\tan(\cdot)$. To address this, we apply an input transformation as in  \cite{karafyllis2022lyapunov} to transform the model into  a control-affine form. Specifically, we set $u_1 = \tan(\delta)$, and we obtain $\delta = \tan^{-1}(u_1)$. For consistency, we also denote $\dot{v} = u_2$. This allows us to express the model as:
\begingroup
\setlength{\belowdisplayskip}{10pt} 
\setlength{\abovedisplayskip}{10pt} 
\begin{align} \label{bicycle_model_2}
    \dot{x} &= v \cos(\theta), \quad \dot{\theta} = \frac{v}{\sigma} u_1, \nonumber \\
    \dot{y} &= v \sin(\theta), \quad \dot{v} = u_2.
\end{align}
\endgroup
\noindent After this transformation, the new control inputs of our system are $u_1$ and $u_2$. To avoid ambiguity with the previous notation we denote the state vector associated with \eqref{bicycle_model_2} as $\textbf{x}_b=[x,y,\theta,v]$, where the subscript $b$ stands for the bicycle model. The new drift vector field is $f(\textbf{x}_b)= [v\text{cos}(\theta),v\text{sin}(\theta),0,0]$ and the new control vector field is $g(\textbf{x}_b)=[0,0,\frac{v}{\sigma}, 1]$. The system can now be written as 
\begin{equation}
\dot{\textbf{x}}_b = f(\textbf{x}_b) + g(\textbf{x}_b)\textbf{u}_b,    
\end{equation}where $\textbf{u}_b = [u_1, u_2]^T$.
\begingroup
\begin{remark}
Hereinafter, we must consider that an input transformation related to the steering angle $\delta$ was applied when deriving any constraints associated with $\delta$.
\end{remark}
In the next subsections, we first define the unsafe set associated with a pedestrian and construct a constraint linking this unsafe set to the coordinates of a CAV (\textit{4.2}). Next, we develop sufficient barrier conditions based on this constraint (\textit{4.3}), and we also enforce desired states according the mode of each CAV (\textit{4.4}). Finally, we present the QP problem that is followed by the CAVs when emergency mode is activated, along with the replanning mechanism adopted by the rest of CAVs (\textit{4.5}-\textit{4.6}).
\subsection{Unsafe set associated with the pedestrian}
While there are advanced models for human-motion prediction; see \cite{bajcsy2020robust}, we choose to develop a conservative, designer-friendly, model to define the unsafe set associated with the pedestrian, for two main reasons. First, no existing data reliably captures human behavior in a fully CAV-integrated transportation network. Second, when safety depends on human behavior, conservative approaches are preferable in the absence of complete certainty. We select the unsafe set to be defined as an ellipse centered around the pedestrian, with the pedestrian positioned at the rear half to reflect the natural dynamics of human motion. This arrangement considers that pedestrians typically move forward rather than backward while the elliptical shape accounts for limited lateral movements, which are generally less expected and physically challenging. The setup of the ellipse is shown in Fig. \ref{CAV_and_ellipse}.

\begin{figure}[h!]
    \centering
        \centering
        \includegraphics[width=0.83\linewidth]{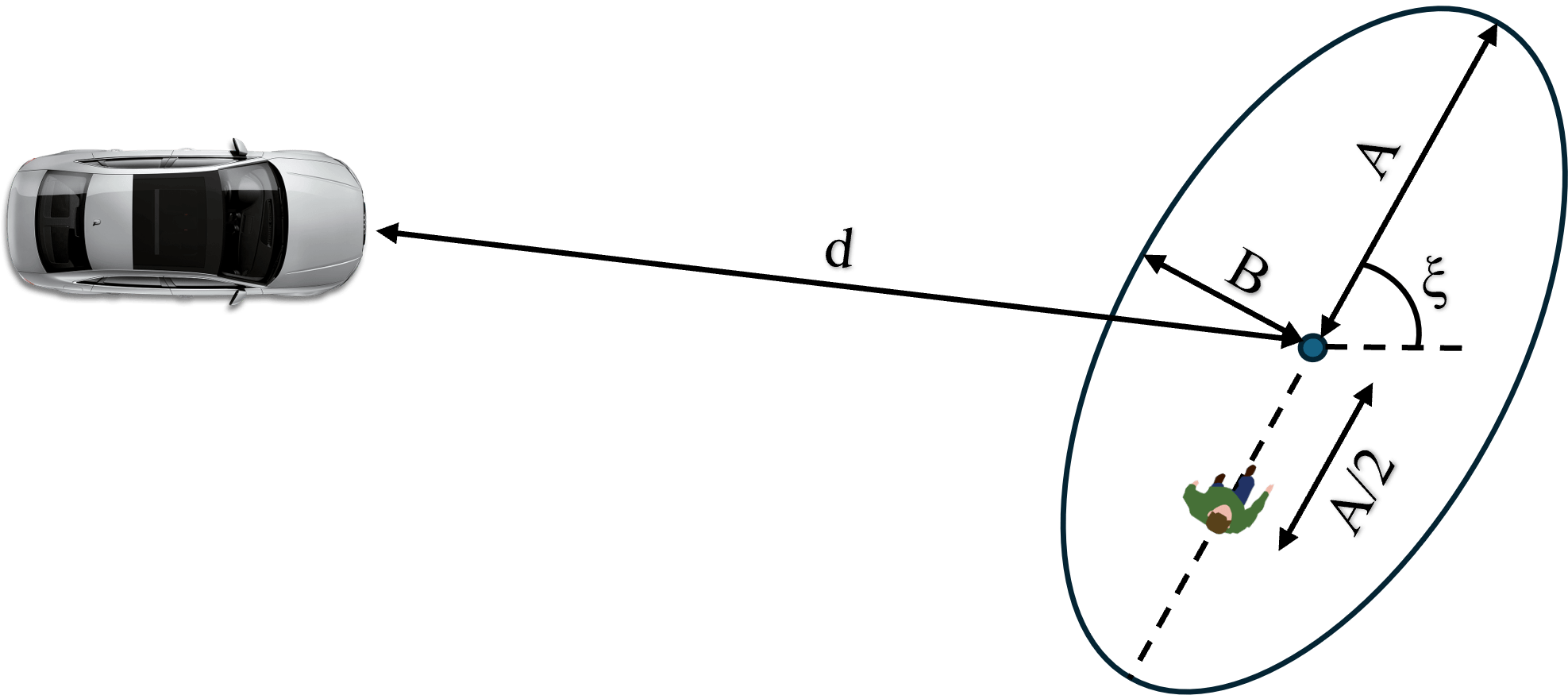}
        \label{bicycle_model}
\caption{Pedestrian surrounded by the ellipse.}
\label{CAV_and_ellipse}
\end{figure}

The range of the ellipse is selected to be a function of 1) pedestrians' speed, 2) pedestrians' orientation, 3) distance between the pedestrian and CAV, and 4) CAV's speed. We justify this selection for the following reasons: 1) the pedestrian's velocity can affect the dimensions of the unsafe set, as inertia causes higher velocities to correspond to a larger range of potential movement, 2) this potential movement is directly correlated to the pedestrian's orientation, 3) since predicting pedestrian intentions is more challenging at greater distances, it is reasonable that the ellipse's range increase with the distance to the CAV and 4) as a slow-moving CAV poses minimal risk, we want the range of the ellipse to be proportional to the vehicle’s speed.  Given these factors, we present the following metric that adjusts the ellipse's range according to the state of a CAV $i$:
\begingroup
\setlength{\belowdisplayskip}{10pt} 
\setlength{\abovedisplayskip}{10pt} 
\begin{align}
    &A(t) = \epsilon + \frac{v_{\text{ped}}(t)}{k_1}\frac{d_i(t)}{k_2}\frac{v_i(t)}{k_3}, \label{ellipse1} \\
    &B(t) = \frac{A(t)}{\lambda}, \label{ellipse2}
    \
\end{align}
\endgroup
where, $A(t)$ and $B(t)$ denote the magnitude of the major and minor axis respectively while $\epsilon$ denotes a safe standstill distance between the CAV and the pedestrian. Also, $v_{\text{ped}}(t)$ and $v_i(t)$ represent the velocities of the pedestrian and CAV $i$ at time $t$, respectively, and $d_i(t)$ is the Euclidean distance between them. The parameters $k_1$, $k_2$, $k_3$, and $\lambda$ are adjustable to set the model's conservativeness and the importance of each term. It is worth noting that if $\lambda = 1$, the shape becomes circular.  Finally, note that the units of speed and distance are excluded in \eqref{ellipse1} to avoid an uninterpretable unit for the ellipses' axes. 

Given the formula for defining the ellipse, we now need to construct the safety constraint between CAV and pedestrian accordingly. Since we also consider the orientation of the walking pedestrian we want a generalized constraint associated with a rotated ellipse that is not centered at the origin. The derivation of this constraint that represents the CBF between CAV and pedestrian is given in the next Theorem.

\begin{theorem}
Let the center of the elliptical unsafe set associated with the pedestrian be denoted by $(x_0, y_0) \in \mathbb{R}^2$, and let the pedestrian's orientation be given by the angle $\xi \in [-\pi, \pi)$. Consider a CAV and let $r_b > 0$ denote the distance from the CAV’s barycentric center to any of its corners. The CBF that encodes safety with respect to the pedestrian’s position and orientation is as follows:
\begin{align} 
    &b_5(\textbf{x}_b)=: \nonumber \\
    &\frac{\left((x +  \frac{\sigma}{2}\text{cos}(\theta) - x_0)\text{cos}(\xi) + (y +  \frac{\sigma}{2}\text{sin}(\theta) - y_0)\text{sin}(\xi)\right)^2}{A^2} + \nonumber \\
    &\frac{\left((x +  \frac{\sigma}{2}\text{cos}(\theta) - x_0)\text{sin}(\xi) - (y +  \frac{\sigma}{2}\text{sin}(\theta) - y_0)\text{cos}(\xi)\right)^2}{B^2} \nonumber \\
    &- 1 - r_b \geq 0 \label{CAV_Pedestrian_Constraint}
\end{align}
\end{theorem}
\begin{proof}
The canonical ellipse centered at the origin and aligned with the coordinate axes is
\begin{equation} \label{canonical ellipse}
\frac{x^{2}}{A^{2}}+\frac{y^{2}}{B^{2}}=1.
\end{equation}
To relocate the centre to $(x_{0},y_{0})$ we translate the coordinates by setting $\tilde{x}=x-x_{0}$ and $\tilde{y}=y-y_{0}$, which yields 
\begin{equation} \label{translated ellipse}
\frac{\tilde{x}^{2}}{A^{2}}+\frac{\tilde{y}^{2}}{B^{2}}=1.
\end{equation}
Next we express $(\tilde{x},\tilde{y})$ in a frame obtained by a passive, counter‑clockwise rotation through the angle~$\xi$.  Denoting the rotated coordinates by $(x',y')$, we have
\[
\begin{pmatrix}x'\\y'\end{pmatrix}
=
R(\xi)
\begin{pmatrix}\tilde{x}\\\tilde{y}\end{pmatrix},
\qquad 
R(\xi)=
\begin{pmatrix}
\cos(\xi)&\sin(\xi)\\
-\sin(\xi)&\cos(\xi)
\end{pmatrix}.
\]
Thus,
\begin{align}
&x'=(x-x_{0})\cos(\xi)+(y-y_{0})\sin(\xi), \nonumber \\
&y'=-(x-x_{0})\sin(\xi)+(y-y_{0})\cos(\xi).\nonumber
\end{align}
Substituting $(x',y')$ into \eqref{translated ellipse} yields
\begin{align}
    &\frac{((x-x_{0})\cos(\xi)+(y-y_{0})\sin(\xi))^{2}}{A^{2}}\nonumber\\
    &\quad+\frac{(-(x-x_{0})\sin(\xi)+(y-y_{0})\cos(\xi))^{2}}{B^{2}}=1\label{equation_of_the_ellipse}.
\end{align}
The equality in \eqref{equation_of_the_ellipse} describes the boundary $\partial\mathcal{E}$ of the ellipse, while the interior (the unsafe region) is 
\[
\mathcal{E}^{\mathrm{int}}=\{(x,y)\in\mathbb{R}^{2}:g(x,y)<1\},
\]
where $g(x,y)$ denotes the left‑hand side of~\eqref{equation_of_the_ellipse}.  Safety therefore requires $g(x,y)\ge 1$. Next, recall that $(x,y)$ denotes the rear‑axle reference point of a vehicle of length $\sigma$ and heading~$\theta$. By applying basic trigonometric principles its barycentric centre is calculated as
\[
x_{c}=x+\tfrac{\sigma}{2}\cos(\theta),\qquad
y_{c}=y+\tfrac{\sigma}{2}\sin(\theta).
\]
To guarantee that every point of the vehicle lies outside $\mathcal{E}^{\mathrm{int}}$, we augment the ellipse by the closed disc $D_{r_{b}}$ of radius~$r_{b}$. Thus the final constraint is
\[
g(x_{c},y_{c})\;\ge\;1-r_b,
\]
 where we deduct $r_{b}$ to ensure that the front, rear and corners of the vehicle remain outside the unsafe set at all times. That completes the proof.
\end{proof}

Thus far, we have derived the bicycle’s kinematic model, which is mathematically expressed in \eqref{bicycle_model_2}, and defined a metric to compute the pedestrian's unsafe set in real time, given by \eqref{ellipse1}-\eqref{ellipse2}. We then introduced CBF \eqref{CAV_Pedestrian_Constraint} linking the unsafe set to the full dimensions of the CAV. That is, we require the coordinates \((x, y)\) of each CAV to satisfy \eqref{CAV_Pedestrian_Constraint}. In the next subsection, we design sufficient barrier conditions based on \eqref{CAV_Pedestrian_Constraint} to ensure the safe operation of the CAVs when a pedestrian is detected.

\subsection{Derivation of barrier conditions}
In comparison to the analysis presented in Section 3, here, if we apply \eqref{CAV_Pedestrian_Constraint} to the inequality \eqref{CBF constraint} (as we did for constraints \eqref{CBFrearend1},  \eqref{Cert_acc_1}, \eqref{Cert_acc_2}, \eqref{Cert_lateral}),  we will not obtain a barrier condition including both the control actions $u_1,u_2$. To circumvent this issue, we introduce the concept of CBFs with higher-order relative degrees that will allow both control actions to arise in the final barrier condition. Before we apply this concept on \eqref{CAV_Pedestrian_Constraint}, we review some preliminary definitions as defined by \cite{xiao2021bridging}.
\begin{definition}
\textbf{(Relative Degree).} The relative degree of a differentiable function $b : \mathbb{R}^n \rightarrow \mathbb{R}$ with respect to some system dynamics, with state vector $\textbf{x}$, is the number of times it is differentiated along its dynamics until all the components of the control vector $\mathbf{u}$ explicitly show in the corresponding derivative.
\end{definition}
For a state vector $\textbf{x}$ and a constraint $b(\textbf{x}) \geq 0$, with relative degree $m$, $b : \mathbb{R}^n \rightarrow \mathbb{R}$, and $\psi_0(\textbf{x}) := b(\textbf{x})$, we define a sequence of functions $\psi_i : \mathbb{R}^n \rightarrow \mathbb{R}$, $i \in \{1, \ldots, m\}$:
\begin{equation} \label{Psi_functions}
    \psi_i(\textbf{x}) := \dot{\psi}_{i-1}(\textbf{x}) + \alpha_i(\psi_{i-1}(\textbf{x})), \quad i \in \{1, \ldots, m\},
\end{equation}
where $\alpha_i(\cdot)$, $i \in \{1, \ldots, m\}$ denotes a $(m - i)^{\text{th}}$ order differentiable class $\mathcal{K}$ function. We further define a sequence of sets $C_i, i \in \{1, \ldots, m\}$ associated with \eqref{Psi_functions} in the form:
\begin{equation}\label{C_sets}
    C_i := \{\textbf{x} \in \mathbb{R}^n : \psi_{i-1}(\textbf{x}) \geq 0\}, \quad i \in \{1, \ldots, m\}.
\end{equation}

\begin{definition}[High Order Control Barrier Function (HOCBF)] \label{Higher_OrderCBF}
Let $C_1, \ldots, C_m$ be defined by \eqref{C_sets} and $\psi_1(\textbf{x}), \ldots, \psi_m(\textbf{x})$ be defined by \eqref{Psi_functions}. A function $b : \mathbb{R}^n \rightarrow \mathbb{R}$ is a HOCBF of relative degree $m$ if there exist $(m - i)^{\text{th}}$ order differentiable class $\mathcal{K}$ functions $\alpha_i, i \in \{1, \ldots, m-1\}$ and a class $\mathcal{K}$ function $\alpha_m$ such that
\begin{align}
    \sup_{\textbf{u} \in U} \bigg( \mathcal{L}_f^m b(\textbf{x}) + &\left[ \mathcal{L}_g L_f^{m-1} b(\textbf{x}) \right] u + \nonumber \\
    &O(b(\textbf{x})) + \alpha_m(\psi_{m-1}(\textbf{x})) \bigg) \geq 0, \label{Higher_Order_Certificate}
\end{align}
$\forall$ $\textbf{x} \in C_1 \cap \ldots \cap C_m$, and $O(\cdot) = \sum_{i=1}^{m-1} \mathcal{L}_f^i (\alpha_{m-i} \circ \psi_{m-i-1}(\textbf{x}))$.
\end{definition}
Following Definition 4 
 and \eqref{Higher_Order_Certificate}, and by taking as class $\mathcal{K}$ functions the identity function, we obtain the following form for our barrier condition as was defined in \cite{xiao2022control} (see (19) and (20) therein) :
\begin{align} \label{2orderCBF}
    \mathcal{L}^2_f b(\textbf{x}) + \mathcal{L}_g\mathcal{L}_f b(\textbf{x}) u_1 &+ \mathcal{L}_g\mathcal{L}_f b(\textbf{x}) u_2 + 2 \mathcal{L}_f b(\textbf{x}) + b(\textbf{x}) \geq 0 .
\end{align}
Hence, the final step is to calculate \eqref{2orderCBF} based on $b_5(\textbf{x}_b)$ that was defined in \eqref{CAV_Pedestrian_Constraint}. The corresponding values for this derivation are presented in the following proposition. For brevity, we substitute $q_1=\frac{\sigma}{2}$, $q_2=(x + q_1\text{cos}(\theta) - x_0)$ and $q_3=(y + q_1\text{sin}(\theta) - y_0)$. 
\begin{proposition}
For $b(\textbf{x})=b_5(\textbf{x}_b)$ in \eqref{2orderCBF}, the terms $b_5(\textbf{x})$, $\mathcal{L}_f b_5(\textbf{x})$, $\mathcal{L}^2_f b_5(\textbf{x})$, $\mathcal{L}_g\mathcal{L}_f b_5(\textbf{x}) u_1$,  $\mathcal{L}_g\mathcal{L}_f b_5(\textbf{x}) u_2$, are given by:
\vspace{-5pt}
\begingroup
\setlength{\belowdisplayskip}{10pt} 
\setlength{\abovedisplayskip}{10pt} 
\begin{align}
     \boldsymbol{b(\textbf{x})}&= \frac{\left(q_2\text{cos}(\xi)+q_3\text{sin}(\xi)\right)^2}{A^2} +\frac{\left(q_2\text{sin}(\xi)-q_3\text{cos}(\xi)\right)^2}{B^2}  \nonumber \\
     & -r_b -1 \nonumber \\
    \boldsymbol{\mathcal{L}_fb}& =v\text{cos}(\theta)\frac{2\text{cos}(\xi)\left(q_2\text{cos}(\xi)+q_3\text{sin}(\xi)\right)}{A^2}  \nonumber \\ & +v\text{cos}(\theta)\frac{2\text{sin}(\xi)\left(q_2\text{sin}(\xi)-q_3\text{cos}(\xi)\right)}{B^2} \nonumber \\
    & +v\text{sin}(\theta)\frac{2\text{sin}(\xi)\left(q_2\text{cos}(\xi)+q_3\text{sin}(\xi)\right)}{A^2}  \nonumber\\
    &-v\text{sin}(\theta)\frac{2\text{cos}(\xi)\left(q_2\text{sin}(\xi)-q_3\text{cos}(\xi)\right)}{B^2}\nonumber \\
    \boldsymbol{\mathcal{L}_f^2b} &= v\text{cos}(\theta)\frac{2\text{cos}^2(\xi)}{A^2}+ v\text{cos}(\theta) \frac{2\text{sin}^2(\xi)}{B^2}  \nonumber \\ 
    & + v\text{sin}(\theta)\frac{2\text{sin}(\xi)\text{cos}(\xi)}{A^2} - v\text{sin}(\theta)\frac{2\text{cos}(\xi)\text{sin}(\xi)}{B^2}    \nonumber \\
    & + v\text{cos}(\theta) \frac{2\text{cos}(\xi)\text{sin}(\xi)}{A^2} - v\text{cos}(\theta)\frac{2\text{sin}(\xi)\text{cos}(\xi)}{B^2}  \nonumber \\
    & + v\text{sin}(\theta)\frac{2\text{sin}^2(\xi)}{A^2} - v\text{sin}(\theta)\frac{-2\text{cos}^2(\xi)}{B^2} \nonumber \\ \boldsymbol{\mathcal{L}_{g_{u_2}}\mathcal{L}_fb} &= u_2 \text{cos}(\theta)\frac{2\text{cos}(\xi)\left(q_2\text{cos}(\xi)+q_3\text{sin}(\xi)\right)}{A^2}  \nonumber \\ &+u_2 \text{cos}(\theta)\frac{2\text{sin}(\xi)\left(q_2\text{sin}(\xi)-q_3\text{cos}(\xi)\right)}{B^2} \nonumber \\
    & + u_2\text{sin}(\theta)\frac{2\text{sin}(\xi)\left(q_2\text{cos}(\xi)+q_3\text{sin}(\xi)\right)}{A^2}  \nonumber\\
    &-u_2\text{sin}(\theta)\frac{-2\text{cos}(\xi)\left(q_2\text{sin}(\xi)-q_3\text{cos}(\xi)\right)}{B^2} \nonumber \\
\boldsymbol{\mathcal{L}_{g_{u_1}}\mathcal{L}_fb}&=u_2\frac{v}{\sigma} \cdot   \nonumber \\
 &-v\text{sin}\theta \frac{2\text{cos}(\xi)\left(q_2\text{cos}(\xi)+q_3\text{sin}(\xi)\right)}{A^2}  \nonumber \\
 & - v\text{sin}\theta\frac{2\text{sin}(\xi)\left(q_2\text{sin}(\xi)-q_3\text{cos}(\xi)\right)}{B^2}+ \nonumber \\
 &v\text{cos}\theta\frac{-2q_1\text{cos}^2(\xi)\text{sin}(\theta) + 2q_1\text{sin}(\xi)\text{cos}(\xi)\text{cos}(\theta)}{A^2}+  \nonumber \\
 &v\text{cos}\theta\frac{-2q_1\text{sin}^2(\xi)\text{sin}(\theta)-2q_1\text{sin}(\xi)\text{cos}(\xi)\text{cos}(\theta)}{B^2} +   \nonumber \\
 &  v\text{cos}\theta\frac{2\text{sin}(\xi)\left(q_2\text{cos}(\xi)+q_3\text{sin}(\xi)\right)}{A^2}-  \nonumber \\
 &v\text{cos}\theta\frac{2\text{cos}(\xi)\left(q_2\text{sin}(\xi)-q_3\text{cos}(\xi)\right)}{B^2}+\nonumber \\
 &v\text{sin}\theta \frac{-2q_1\text{sin}(\xi)\text{cos}(\xi)\text{sin}(\theta) + 2q_1\text{sin}^2(\xi)\text{cos}(\theta)}{A^2}-  \nonumber \\
 &v\text{sin}\theta\frac{-2q_1\text{cos}(\xi)\text{sin}(\xi)\text{sin}(\theta)-2q_1\text{cos}^2(\xi)\text{cos}(\theta)}{B^2}. \label{Pedestrian_Certificate}
 \end{align} 
 \endgroup
\end{proposition}
\begingroup
\begin{proof}
The analytical derivation follows the process outlined in Remark 3. Due to the intricacy of the algebraic analysis involved, it is excluded for brevity.
\end{proof}
\endgroup
\begingroup

\begin{remark}
Because of space constraints we cannot list every term of \eqref{Pedestrian_Certificate} in a single barrier condition. Throughout the paper we slightly abuse notation by denoting the barrier condition associated with \eqref{CAV_Pedestrian_Constraint} simply as \eqref{Pedestrian_Certificate}.
\end{remark}

So far we have derived a barrier condition that couples the pedestrian’s unsafe set with the CAV’s control inputs. 
Because the steering angle is now an explicit state, its constraints—as well as those arising from the lane boundaries—must also be enforced. 
We therefore introduce additional CBFs to guarantee safety with respect to both the road edges and the steering‑angle limits.

The derivation of constraints associated with the lane markings is relatively straightforward.  Namely, when moving along the y-axis, we define the left boundary as \( x_{\text{left}} \) and the right boundary as \( x_{\text{right}} \). Hence we obtain the constraint \( x_{\text{left}} \leq x \leq x_{\text{right}} \), which can be expressed as \( b_6(\textbf{x}_b) = x - x_{\text{left}} \geq 0 \) and \( b_7(\textbf{x}_b) = x_{\text{right}} - x \geq 0 \). Similarly, for the y axis, we obtain $b_8(\textbf{x}_b)=y-y_{\text{right}}\geq0$ and $b_9(\textbf{x}_b)=y_{\text{left}}-y\geq0$. Yet, when a path is curved we use a combination of $b_6, b_7, b_8$ and $b_9$. By applying \eqref{2orderCBF}, results to the following barrier conditions:
\begin{align} \label{respecting_the_road}
    &0-v\text{sin}(\theta)u_1+\text{cos}(\theta)u_2+2v\text{cos}(\theta)+x-x_{\text{left}}\geq0, \nonumber \\
        &0+v\text{sin}(\theta)u_1-\text{cos}(\theta)u_2-2v\text{cos}(\theta)+x_{\text{right}}-x\geq0, \nonumber \\
        &0+v\text{cos}(\theta)u_1 + \text{sin}(\theta)u_2 + 2v\text{sin}(\theta) + y - y_{\text{right}} \geq0, \nonumber \\
                &0-v\text{cos}(\theta)u_1 - \text{sin}(\theta)u_2 - 2v\text{sin}(\theta) + y_{\text{left}} - y \geq0.
\end{align}
\endgroup
Note that the first term of each inequality is zero because the first term of \eqref{2orderCBF} is zero in each case; that is, $\mathcal{L}^2_f b_i(\textbf{x}_b) = 0$ for $i \in \{6,7,8,9\}$. Next, we want to also consider a physical constraint related to the steering angle $\delta$ which is generated by the centripetal force of each CAV. Given that the steering angle $\delta$ is already a control input, the associated constraint is already in linear form with respect to $\delta$. To define this constraint, we take motivation from an idea developed by \cite{hsu2009estimation} who estimates the steering limit based on the speed. Especially, the formula we use to get the maximum steering angle according to the speed of the vehicle is:
\begin{align} \label{steering_angle_constraint}
    &\delta_{\text{max}} = |\delta_{\text{max}}(0)(1-\frac{v}{v_{\text{max}}})|.
\end{align}
\noindent Here, $\delta_{\text{max}}(0)$ refers to the maximum possible angle of the steering wheel when the vehicle is idle. Note that in \eqref{bicycle_model_2} we have applied an input transformation, and consequently, we must follow this transformation for the steering angle constraint. In the following proposition we present the final constraint.
\begin{proposition}
Given the applied input transformation $u_1 = \tan(\delta)$, the constraint related to the steering angle is: 
\begin{align}
|u_{1}| \leq \text{tan}(|\delta_{\text{max}}(0)(1-\frac{v}{v_{max}})|). \label{Steering_Constraint}
\end{align}
\end{proposition}
\begin{proof}
We know from the applied input transformation that $\delta = \tan^{-1}(u_1)$. Hence, the constraint in \eqref{steering_angle_constraint} becomes:
$\tan^{-1}(u_1) \leq |\delta_{\text{max}}(0)(1 - \frac{v}{v_{\text{max}}})|.$
By applying the $\tan$ function (that is increasing) to both sides, we obtain
$|u_1| \leq \tan\left(|\delta_{\text{max}}(0)(1 - \frac{v}{v_{\text{max}}})|\right).$
\end{proof}

Finally, to consider abrupt changes in acceleration when a CAV enters an emergency mode, we add a numerical constraint on the jerk to be within a minimum $j_{\text{min}}$ and maximum $j_{\text{max}}$ value, i.e.,
\begin{equation} \label{Jerk_Constraint}
    j_{\text{min}} \leq \frac{u_2(t + \Delta t) - u_2(t)}{\Delta t} \leq j_{\text{max}}.
\end{equation}
Until this point, we derived sufficient CBFs and constraints that each CAV must obey when the emergency mode is activated. However, we have not yet addressed how a CAV smoothly returns to the center of its lane (nominal path) after deviating from it during an emergency maneuver. Additionally, while the certificates account for braking based on proximity to the pedestrian, we want to impose a lower desired speed during emergency mode to ensure uniform speed reduction, even when CAVs are far from the pedestrian and braking is not immediately triggered from the barrier conditions. In the next subsection, we address these cases and incentivize CAVs to: 1) re-align with the center of their lane after safely passing a pedestrian or after the pedestrian clears the roadway, and 2) adopt a lower desired speed during emergency mode.

\subsection{Introducing desired states as soft constraints}
\vspace{-5pt}
\noindent 
After a CAV has left the lane center to avoid a pedestrian and the avoidance maneuver is complete, we define a \textit{recovery period} during which the vehicle realigns itself with the center of the road. Thus, we aim the control inputs to minimize a penalty function that captures the distance of the CAV from the center of the lane. This can be done by using a soft constraint associated with state variables in the desired reference trajectory. Namely, we define the function $E(\textbf{x}_b(t))=(x-x^{ref})^2+(y-y^{ref})^2$ that captures the deviation of the nominal path from the current position. Then, the soft constraint is defined as:
\begin{align} \label{softcos}
    \mathcal{L}^2_f E(\textbf{x}_b) + \mathcal{L}_g\mathcal{L}_f E(\textbf{x}_b) u_1 &+ \mathcal{L}_g\mathcal{L}_f E(\textbf{x}_b) u_2 + \nonumber \\ & 2 \mathcal{L}_f E(\textbf{x}_b) +  E(\textbf{x}_b) \leq e(t),
\end{align}
where $e(t)$ represents a slack variable. To provide an insight into how this works, we will later see that once a CAV exits the emergency mode, it tries to minimize the objective $e^2(t)$ subject to \eqref{softcos}. Hence, the system indirectly drives the left-hand side of \eqref{softcos} towards zero. This incentivizes the vehicles to align with the center of their lane, as deviations increase the magnitude of $e(t)$, which the optimization penalizes.
By expanding \eqref{softcos}, we obtain:
\begin{align} \label{refpath}
    &\underbrace{2v^2}_{\mathcal{L}^2_f E} + \underbrace{u_1 \frac{v2((x^{ref}-x)v\text{sin}(\theta)+(y-y^{ref})v\text{cos}(\theta))}{\sigma}}_{\mathcal{L}_g\mathcal{L}_f E u_1} + \nonumber \\ 
    & \underbrace{u_22((x-x^{ref})\text{cos}(\theta) + (y-y^{ref})\text{sin}(\theta))}_{\mathcal{L}_g\mathcal{L}_f E u_2 } + \nonumber \\
    & \underbrace{4((x-x^{ref})v\text{cos}(\theta) + (y-y^{ref})v\text{sin}(\theta))}_{2\mathcal{L}_f E} + \nonumber \\
    & \underbrace{(x-x^{ref})^2+(y-y^{ref})^2}_{ E} \leq e(t).
\end{align}

\begin{remark}
 When a CAV attempts to return to the center of the lane by using the soft constraint in \eqref{refpath}, it may over-correct when reaching the center of the lane. This can lead to oscillations around the lane center, a common issue in control applications when aiming to reach a desired state. To address this, we apply an additional barrier condition in the case when a CAV has deviated from its initial path. For example, this extra barrier condition ensures that if the CAV has steered to the right, it will not surpass the lane center when returning, thus avoiding overcorrection. This additional barrier condition can be easily computed, similarly to those in \eqref{respecting_the_road}, by replacing the boundary values with the lane center shifted by a tiny value. This addition promotes a smooth transition back to the initial path, as we will thoroughly discuss in the simulation results. 
\end{remark}
By following the same idea we can also add a function $S(\textbf{x}_b)$ that captures deviations from a desired speed. This function is $S(\textbf{x}_b)=(v(t)-v_{\text{ref}})^2$. Accordingly, we obtain: 
\begin{equation} \label{softcos2}
\mathcal{L}_f S(\textbf{x}_b) + \mathcal{L}_g\mathcal{L}_f S(\textbf{x}_b) u_1 + \mathcal{L}_g\mathcal{L}_f S(\textbf{x}_b) u_2 +  S(\textbf{x}_b) \leq s(t) 
\end{equation} 
and by expanding \eqref{softcos2} yields:
\begin{align}
 \underbrace{0}_{\mathcal{L}_g\mathcal{L}_f S u_1}+\underbrace{0}_{\mathcal{L}_g\mathcal{L}_f S u_2} + \underbrace{2(v(t)-v_{\text{ref}})u_2}_{\mathcal{L}_f S} + \underbrace{(v(t)-v_{\text{ref}})^2}_{S}\leq s(t).\label{soft_speed}
\end{align}
\begin{remark}
We observe that \eqref{softcos} and \eqref{softcos2} share a similar structure with the inequalities \eqref{CBF constraint} and \eqref{2orderCBF}, respectively. First, note that the direction of the inequalities in \eqref{softcos} and \eqref{softcos2} differs from that in \eqref{CBF constraint} and \eqref{2orderCBF}. This difference arises because the inequalities in \eqref{softcos} and \eqref{softcos2} are treated as soft constraints, aiming to drive the left-hand side towards zero. Secondly, note that for \eqref{softcos2}, we built upon \eqref{CBF constraint}, as it was unnecessary to compute higher-order derivatives. That is, the first derivative of $v_i(t)$ already involves the control input \( u_2 \), which directly governs the speed. In contrast, for \eqref{softcos}, higher-order derivatives were required to ensure that both \( u_1 \) and \( u_2 \) appear explicitly. Thus, we designed the soft constraint using the concept of higher order relative degree; see \eqref{2orderCBF}.
\end{remark}

Having established the system dynamics and the relevant hard and soft constraints, we now formalize the optimization problem followed by the CAVs in emergency mode.
\vspace{-5pt}
\subsection{Optimization problem for pedestrian avoidance}
\vspace{-6pt}
Once a CAV enters an emergency mode, its control inputs are defined through the following QP optimization problem:
\begin{align} 
\underset{u_1,u_2}{\min} \int_{t_{ic}}^{t_{fc}}  &w_1 (u_1-u_{\text{ref}})^2 + w_2 u_2^2 + w_3e(t)^2 + w_4 s(t)^2 dt \nonumber \\
     \label{emergency_mode} \\
    \text{subject to} \quad & \eqref{CBFrearend1}, \eqref{Cert_acc_1},\eqref{Cert_acc_2},\eqref{Cert_lateral}, \eqref{Pedestrian_Certificate} \nonumber \\ &\eqref{respecting_the_road},\eqref{Steering_Constraint},
\eqref{Jerk_Constraint},\eqref{refpath},\eqref{soft_speed}. \nonumber 
\end{align} 
Here, $t_{ic}$ denotes the time when the CAV enters emergency mode, and $t_{fc}$ is the time it fully realigns with the center of its lane (i.e., the nominal path). As already discussed, during the interval $[t_{ic}, t_{fc}]$, the CAV may transition to a recovery period after avoiding the pedestrian. This phase, referred to as recovery mode, involves realigning the vehicle with the center of the road, right after the emergency mode ends. Both emergency and recovery modes are governed by \eqref{emergency_mode}, while the weights $w_i, i\in\{1,2,3,4\}$ are adjusted according to the specific requirements of each mode, as discussed below; see the flow chart in Fig. \ref{fig:flow_chart}.

Next, we discuss in detail how the emergency and recovery modes are governed by \eqref{emergency_mode}. The first two terms in \eqref{emergency_mode} aim to minimize deviations in acceleration from a reference value and steering effort, respectively. The remaining two terms prioritize tracking the nominal path and maintaining the desired speed. During emergency mode, safety is paramount so we set minimal values for \( w_1, w_2 \), and \( w_3 \). This ensures the optimization focuses primarily on satisfying safety constraints and achieving a significantly reduced desired speed. After emergency mode ends (the pedestrian was avoided, or the pedestrian left the road), the CAV may require time to realign during recovery mode. At this stage, increasing \( w_1, w_2 \), and particularly \( w_3 \) is crucial, as this emphasizes control inputs that guide the vehicle back to the lane center. Note that the second term in \eqref{emergency_mode} is included to penalize high steering inputs, which could result in an uncomfortable driving experience.

It is also worth-noting that during emergency mode, the selection of the reference control input \( u_{\text{ref}} \) is not critical, as its weight $w_1$ is assigned a minimal value. This is reasonable given that the vehicle prioritizes avoiding the pedestrian, rather than following an optimal trajectory. On the other hand, during recovery mode, the value of \( u_{\text{ref}} \) is determined based on the analysis presented in Section 3 and the \textit{Potryagin's Minimum Principle}. Note that during the time that a CAV re-aligns with the center of road, constant replanning is applied. Replanning involves that a CAV updates its reference control input \( u_{\text{ref}} \)  considering the states of the other CAV. Also, when defining the trajectories of the other vehicles, the lateral movement of the CAV in recovery mode must be accounted for. Specifically, its trajectory is accounted as moving from its current position to the center of the road, following a path analogous to the hypotenuse of a triangle, as discussed in \cite{Malikopoulos2020}; see Remark 3. Next, we proceed to discuss the functionality of our re-planning mechanism, when a pedestrian is detected.


\begin{figure}
        \centering
        \includegraphics[width=0.45\textwidth]{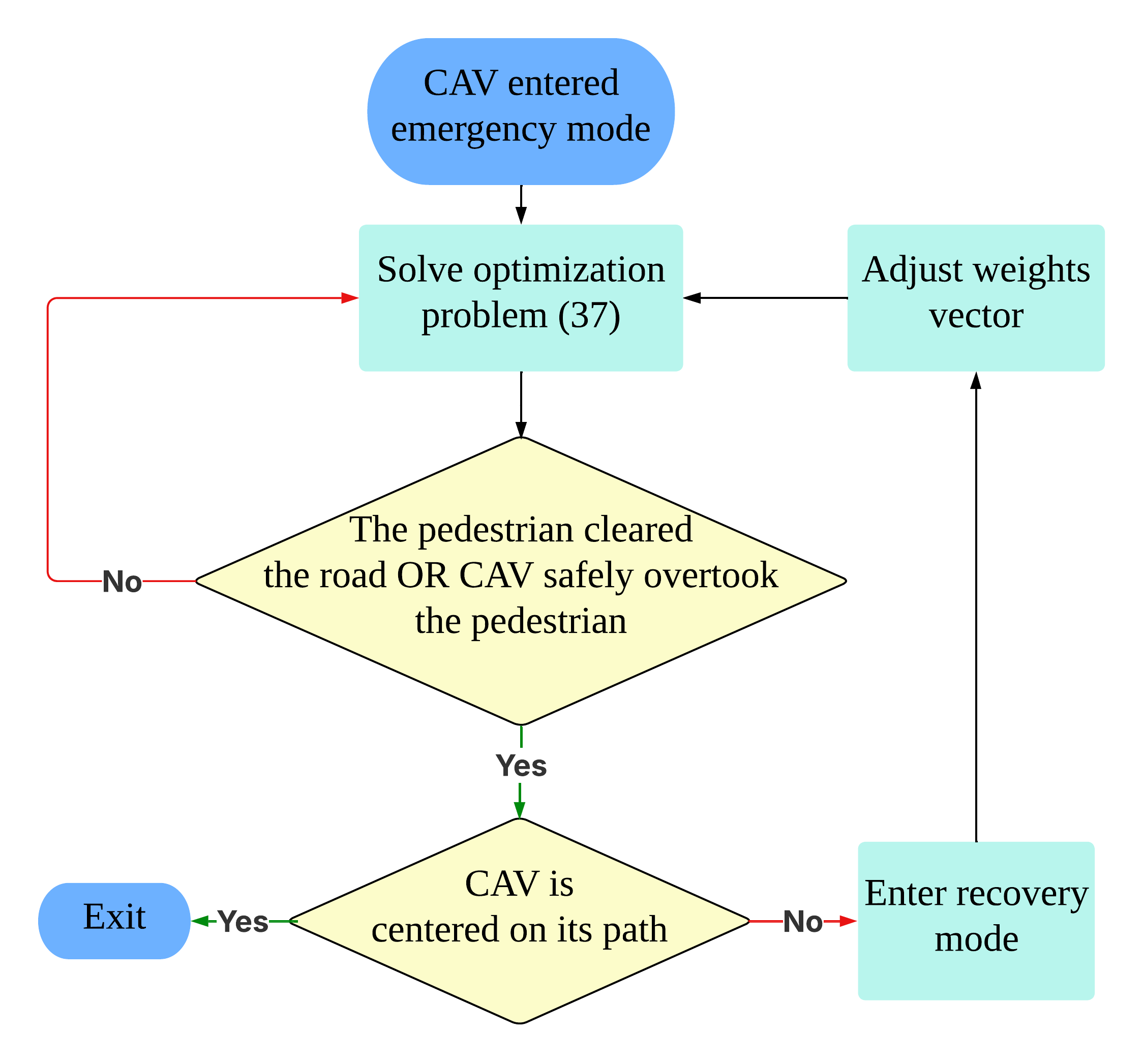}
        \vspace{-5pt}
        \caption{Flow chart when a CAV enters the emergency mode.}
    \label{fig:flow_chart}
\end{figure}
  
\subsection{Replanning mechanism}
\vspace{-6pt}
When a pedestrian is detected, all CAVs in the intersection constantly replan their reference trajectories, so they can account for the states of the CAVs in emergency mode. When replanning occurs, all vehicles must redefine their reference trajectories simultaneously. That requires careful management of the sequence in which they replan. Re-sequencing algorithms, such as the one proposed in \cite{chalaki2021Reseq}, address this issue. In this paper, we utilize the re-sequencing and replanning algorithm introduced in \cite{chalaki2021Reseq}. In this approach, each vehicle is assigned a weight factor $\omega_i$, which is inversely proportional to the length of its interval \( \mathcal{F}_i(t_i^0) = [\underline{t}_i^f, \overline{t}_i^f] \), as defined in the \textit{Time-optimal control problem} in Section 2. This means that vehicles with larger intervals are assigned lower priority, as they have a longer time horizon available for planning. Also for each CAV \( i \in \mathcal{N} \), we define a processing time as \( P_i = \underline{t}_i^f \), representing the minimum time required to exit the control zone. We then categorize the vehicles operating on each path as distinct chains. For each chain, a constraint is imposed such that a vehicle cannot replan before the vehicle preceding it in the chain. Each chain is then associated with a \( \rho \)-factor, calculated as $\rho(z) = \max_{\alpha \in \{1, \dots, k\}} \left( \frac{\sum_{j=1}^{\alpha} \omega_i}{\sum_{j=1}^{\alpha} P_j} \right) = \frac{\sum_{j=1}^{\alpha^*} \omega_i}{\sum_{j=1}^{\alpha^*} P_j}$ where $k$ is the last vehicle on the chain (or path) $z$. The chain with the highest \( \rho \)-factor is given priority. Vehicles in this chain are sequenced first up to the vehicle \( \alpha^* \), which is the CAV that maximizes the \( \rho \)-factor of the chain. These vehicles are then removed from the list and the process is repeated for the remaining CAVs. For a detailed analysis of this algorithm, see \textit{Algorithm 1} in \cite{chalaki2021Reseq}.

\section{Simulations}
\vspace{-10pt}
To demonstrate our framework, we conducted simulations in Matlab. We consider an intersection as shown in Fig. \ref{fig:Intersection}, with a control zone range of 100 meters. Additionally, we consider an object shown in Fig. \ref{fig:pedestrian_detection} that obstructs pedestrian's view of the oncoming traffic, as well as a pedestrian walking behind the object towards the road without noticing approaching vehicles. We assume the following parameter values: $\Delta t= 0.025$s, $\sigma = 2$ m, $u_{\text{max}}=5\ \text{m/s}^2$, $u_{\text{min}} = -5\ \text{m/s}^2$
, $v_\text{min}=0.1\text{m/s}$, $v_\text{max}=25\text{m/s}$, $\phi=1.8$, $\gamma=1.5\text{m}$, $j_{\text{min}}=-7$m/s$^3$, $j_{\text{max}}=5$m/s$^3$, $w=[0,0,0,1]$ in emergency mode and $w=[1.5,1,2,0]$ in recovery mode. We conducted simulations considering two different scenarios with varying initial conditions and parameter settings.

\textit{Scenario 1:} Here, in Fig. \ref{fig:pedestrian_detection} we see the exact moment when a pedestrian enters the right lane of the upward road. Instantaneously, at $t = 4.7$ seconds, the pedestrian is detected by a CAV driving on the same lane with a distance of $18$ meters. In addition, two other CAVs are approaching from behind in the middle and left lanes. For clarity, we label the CAV on the right lane as CAV 1, the one in the middle lane as CAV 2, and the one in the left lane as CAV 3 (see the red vehicles in Fig. \ref{fig:pedestrian_detection}). Immediately after detection, CAV 1 enters an emergency mode and broadcasts the information to the other CAVs, causing CAVs 2, and 3 to enter emergency mode, as well. We validate that by observing the speed profiles in Fig. \ref{fig:speeds}, where CAVs 2 and 3 also reduce their speeds to approach the emergency speed of \( 6 \, \text{m/s} \). In contrast, CAV 1 reduces its speed to \( 2.8 \, \text{m/s} \) due to the pedestrian's proximity. When the pedestrian leaves the road at \( t = 7.1 \) sec, as shown in Fig. \ref{fig: snapshot 3}, CAVs 1, 2, and 3, exit the emergency mode and replan their trajectories using the framework outlined in Section 4.4. In addition, CAV 1 enters a recovery mode due to the heading angle it obtained to avoid the pedestrian. The exit from emergency mode is evident in Fig. \ref{fig:speeds}, where CAVs 1, 2, and 3 adjust their speed profiles, increasing their speeds based on the methods presented in Section 3. Note that the new speed profiles do not attempt to replicate those prior to the critical event, as the vehicles’ states differ during replanning, resulting in unique sequences and trajectories. In addition, it is worth noting that for this scenario, we intentionally selected relatively low initial vehicle speeds to clearly observe changes in the speed profiles before, during, and after the critical event.
Moreover,  in Fig. \ref{fig:minimum_distances}, we observe the minimum rear-end distance among all vehicles and the minimum relative distance from shared conflict points between CAVs during and following the critical event. That confirms that the replanning mechanism effectively provides trajectories that allow all CAVs to safely navigate the intersection.

\begin{figure*} 
        \centering
    \begin{subfigure}[b]{0.24\textwidth}
        \includegraphics[width=\textwidth]{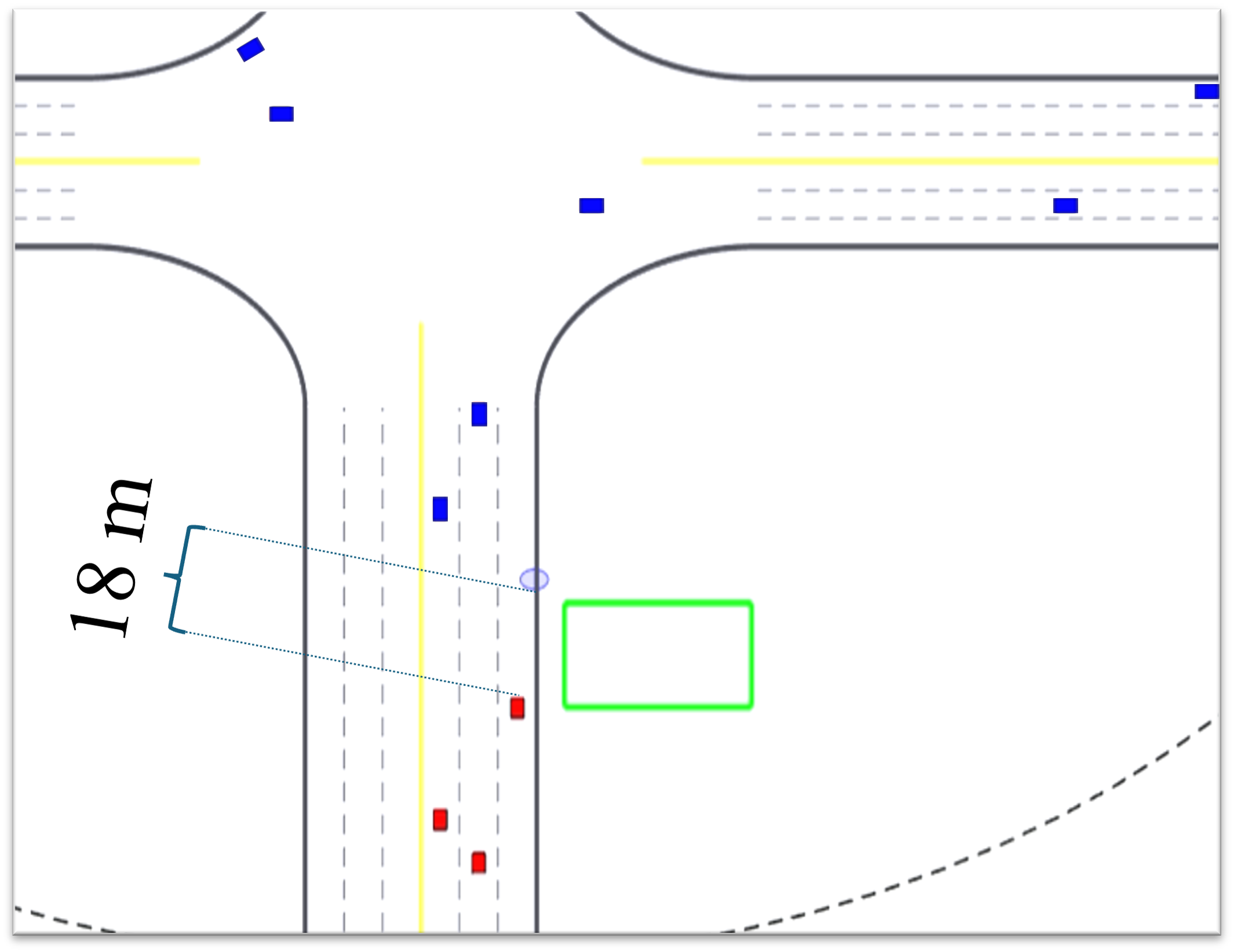}
        \caption{Pedestrian detection}
        \label{fig:pedestrian_detection}
    \end{subfigure}
    \hfill
    \begin{subfigure}[b]{0.24\textwidth}
        \includegraphics[width=\textwidth]{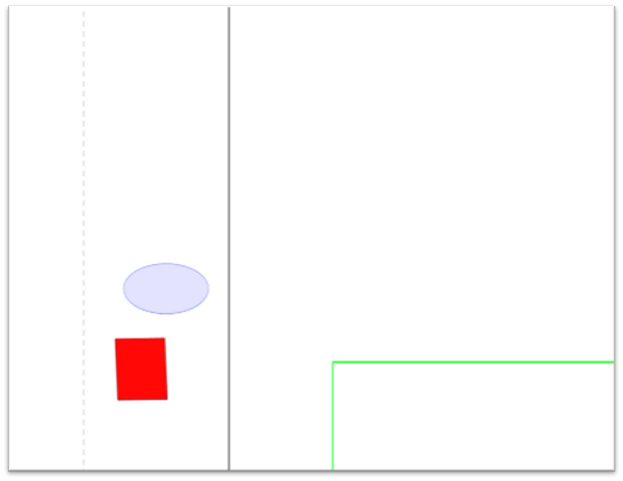}
        \caption{Snapshot 1}
        \label{fig: snapshot 1}
    \end{subfigure}
    \hfill    
    \begin{subfigure}[b]{0.24\textwidth}
        \includegraphics[width=\textwidth]{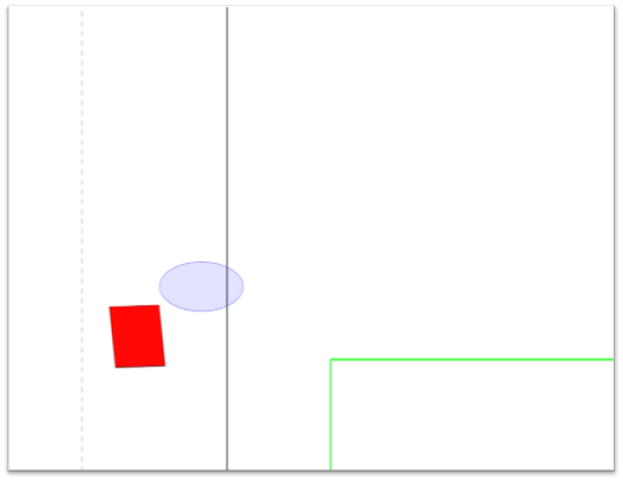}
        \caption{Snapshot 2}
        \label{fig: snapshot 2}
    \end{subfigure}
        \hfill    
    \begin{subfigure}[b]{0.24\textwidth}
        \includegraphics[width=\textwidth]{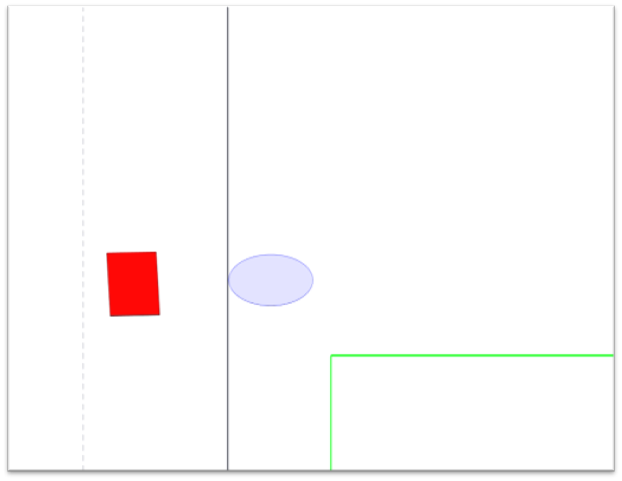}
        \caption{Snapshot 3}
        \label{fig: snapshot 3}
    \end{subfigure}
    \label{fig:Snapshots}
    \caption{Maneuver of CAV for pedestrian avoidance in Scenario 1.}
\end{figure*}

\begin{figure*}

        \centering
    \begin{subfigure}[b]{0.32\textwidth}
        \includegraphics[width=\textwidth]{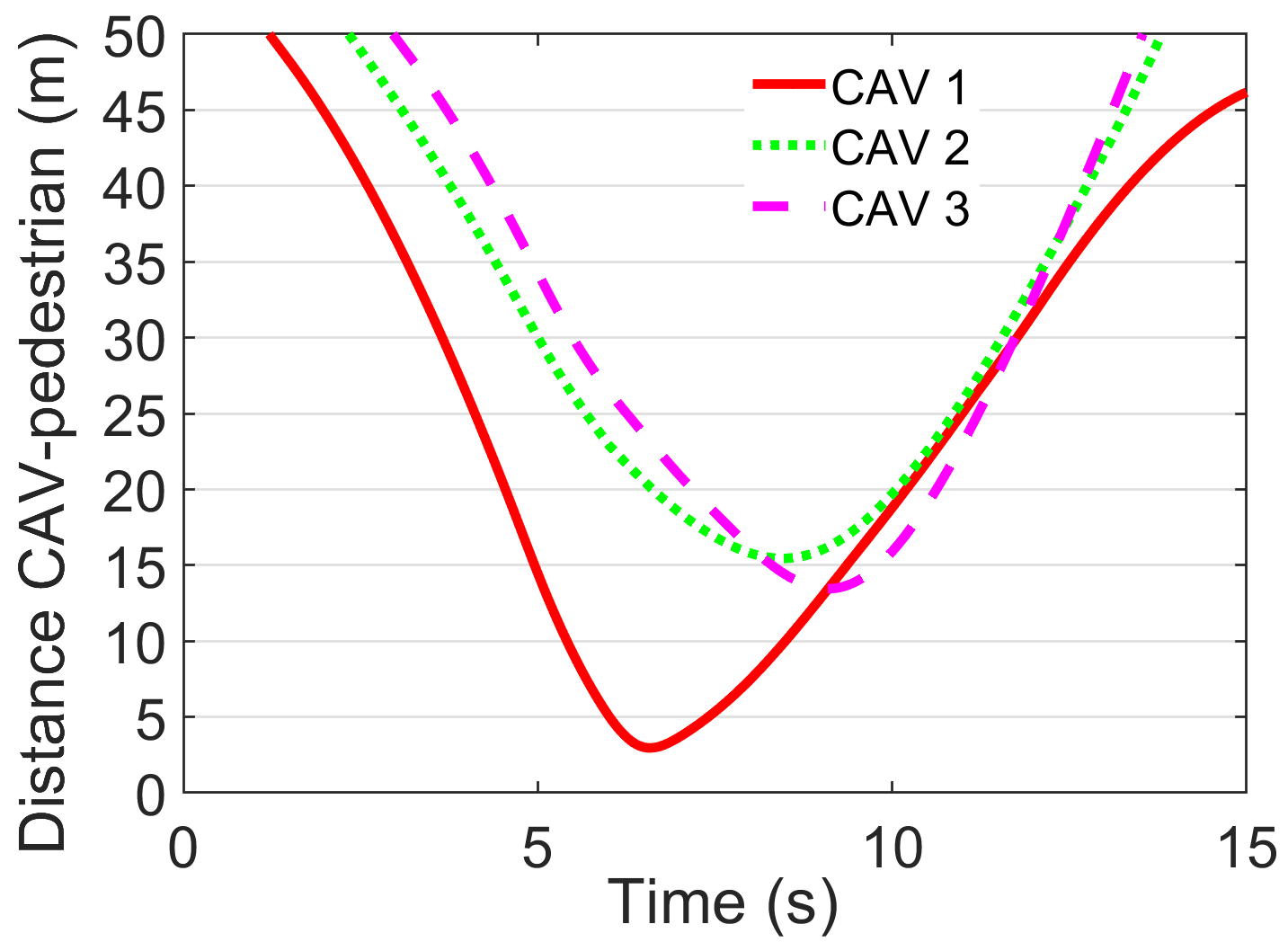}
                \vspace{-15pt}
        \caption{Distance between CAVs - pedestrian}
        \label{fig:distances}
    \end{subfigure}
    \hfill
    \begin{subfigure}[b]{0.32\textwidth}
        \includegraphics[width=\textwidth]{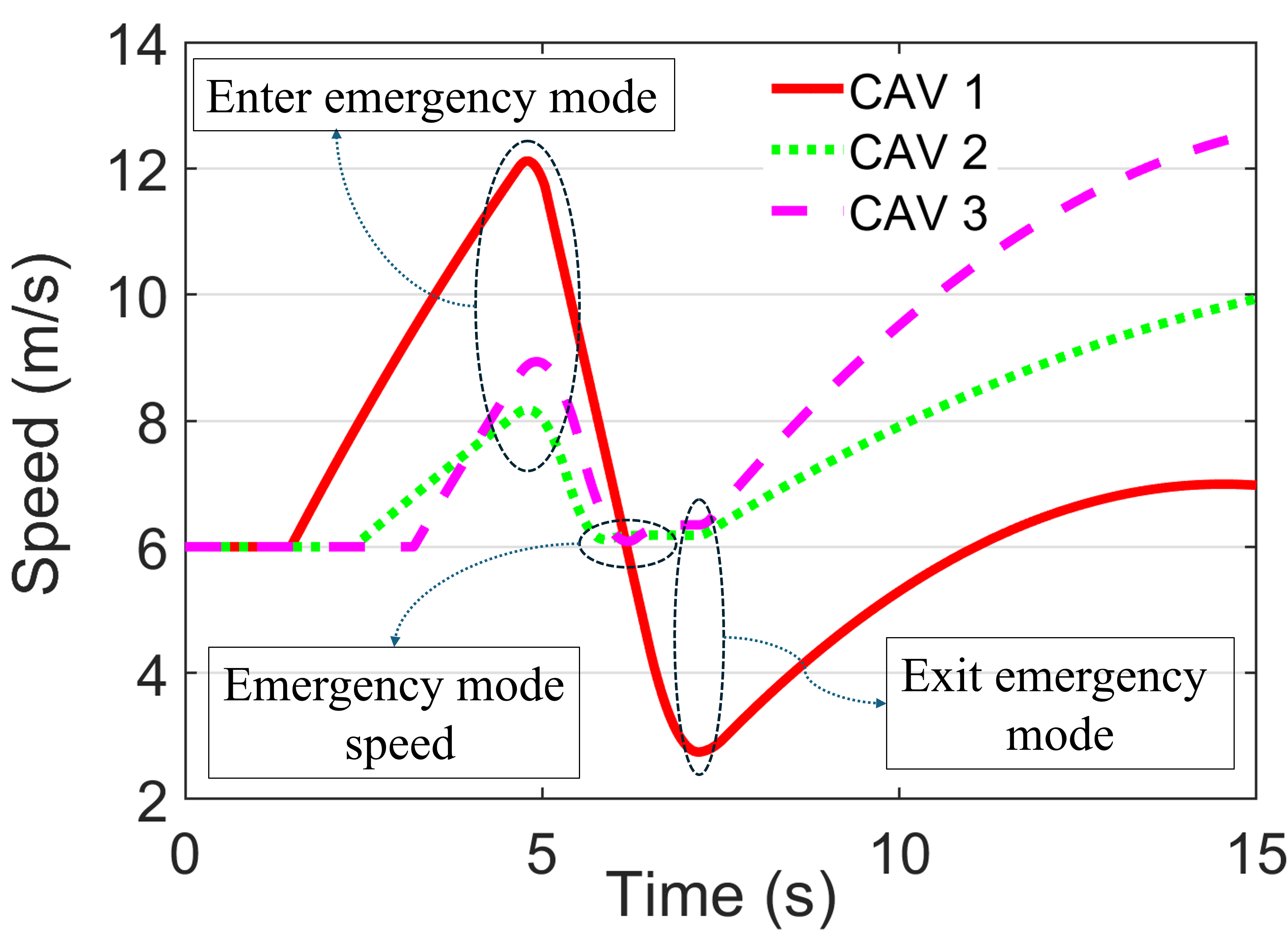}
                \vspace{-15pt}
        \caption{Speeds of CAVs}
        \label{fig:speeds}
    \end{subfigure}
    \hfill    
    \begin{subfigure}[b]{0.32\textwidth}
        \includegraphics[width=\textwidth]{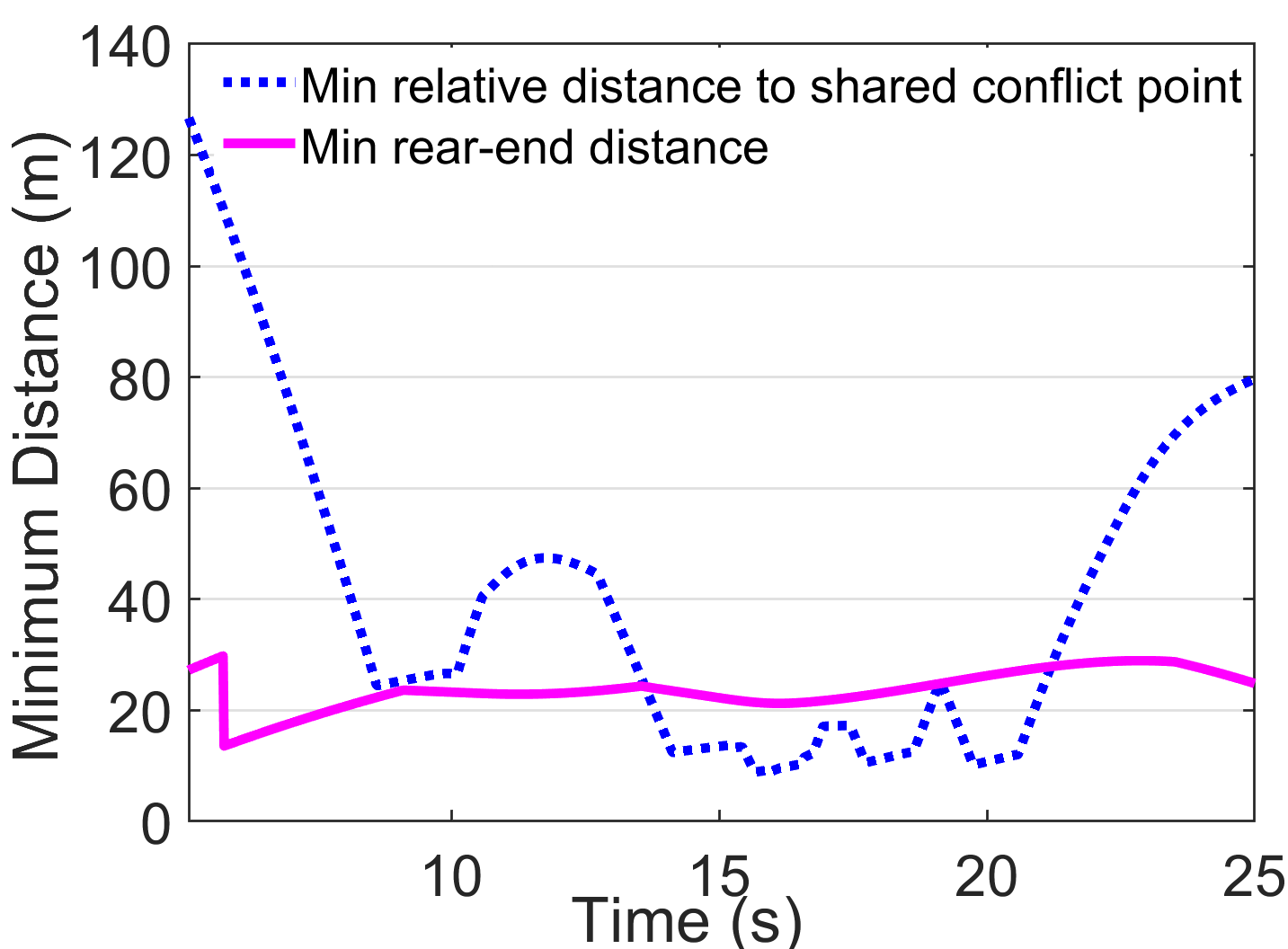}
        \vspace{-15pt}
        \caption{Minimum safety distances}
        \label{fig:minimum_distances}
    \end{subfigure}
    \label{fig:VehicleDynamicsPath5}
    \vspace{-5pt}
    \caption{Verification of safety between CAVs and Pedestrian in Scenario 1}
\end{figure*}

\begin{figure*}
        \centering
    \begin{subfigure}[b]{0.33\textwidth}
        \includegraphics[width=\textwidth]{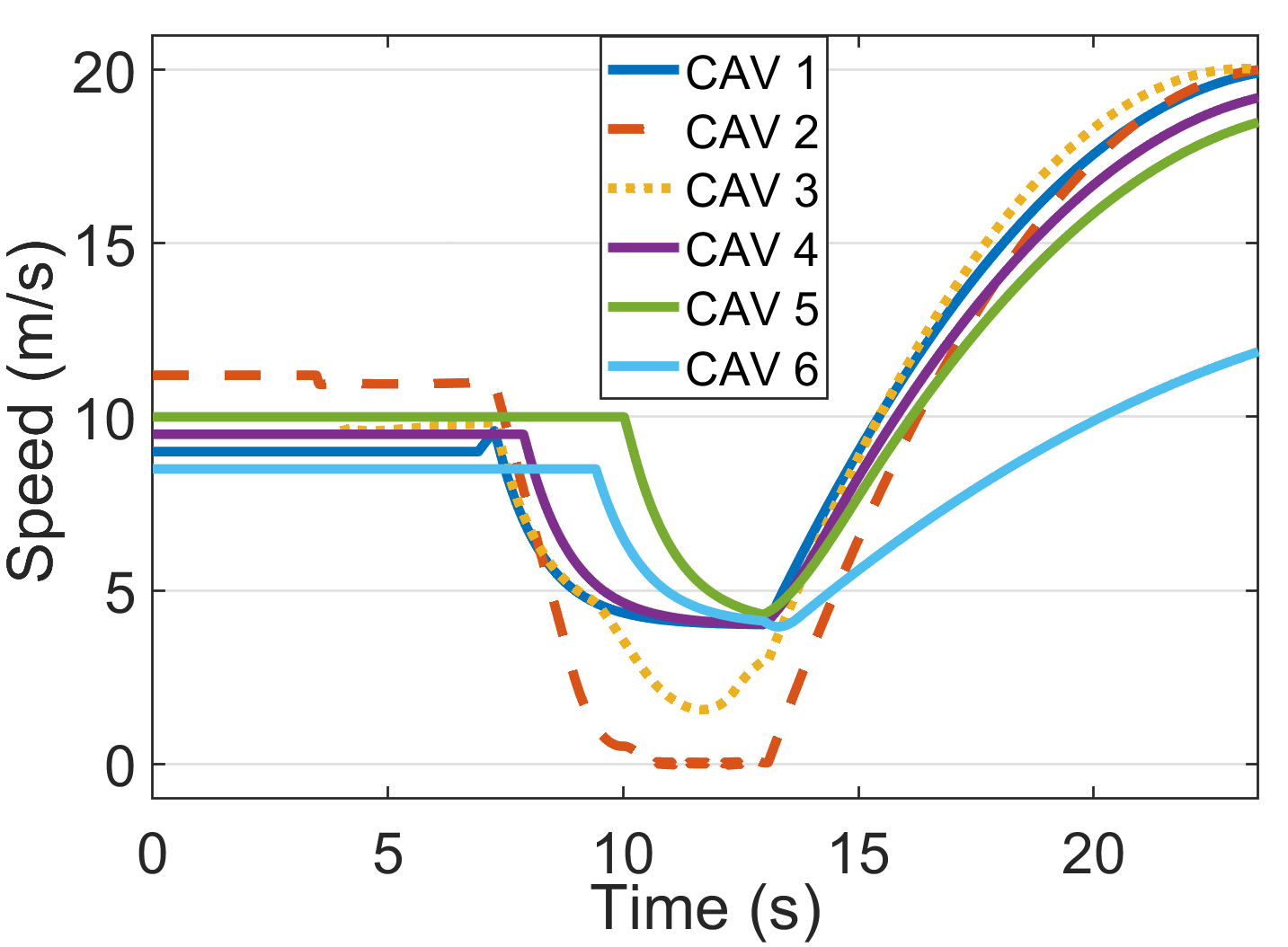}
        \caption{Speeds of CAVs}
        \label{fig:Speeds_2}
    \end{subfigure}
    \hfill
    \begin{subfigure}[b]{0.33\textwidth}
        \includegraphics[width=\textwidth]{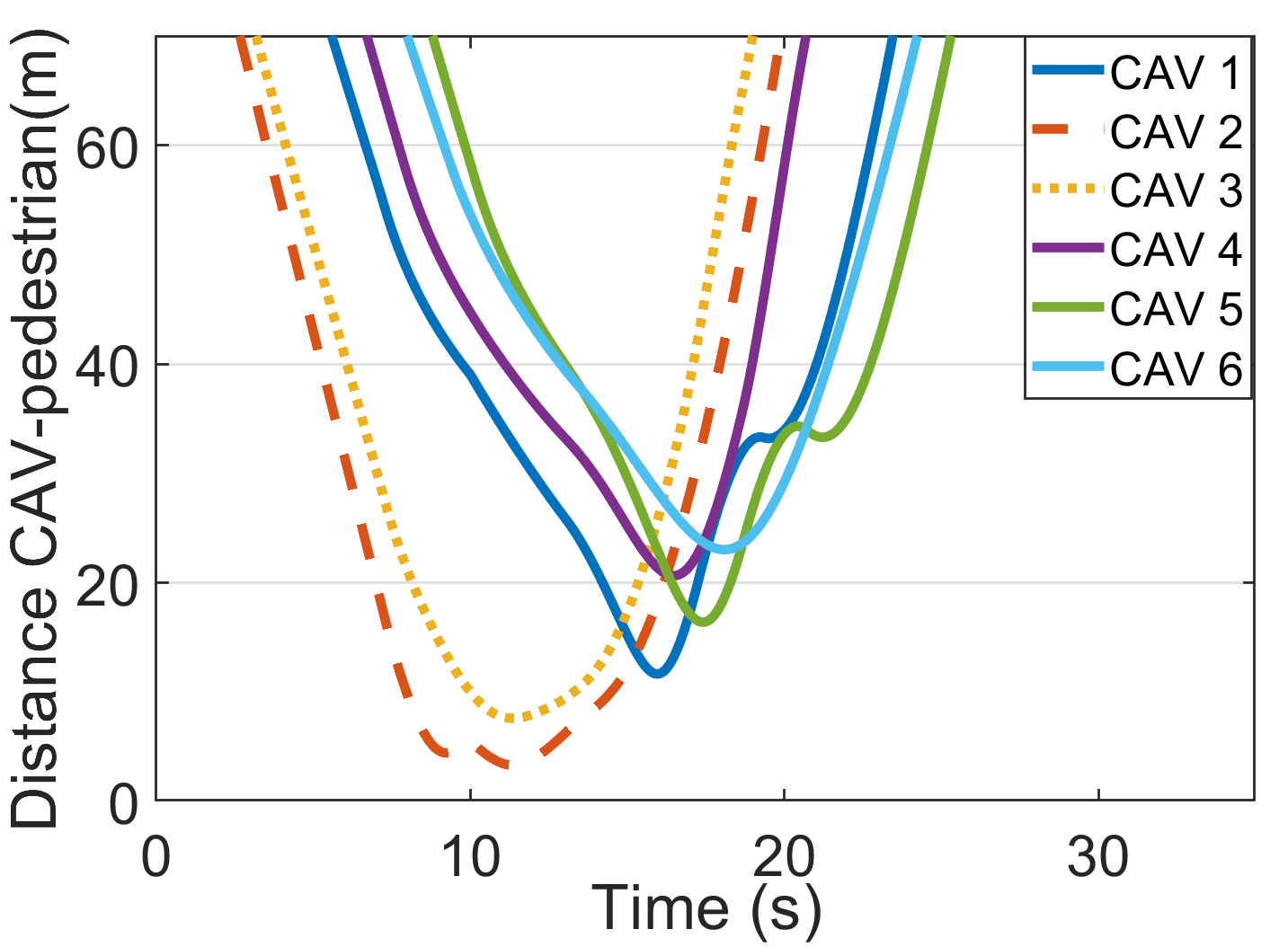}
        \caption{Distance from CAVs}
        \label{fig:Distances_2}
    \end{subfigure}
        \hfill
    \begin{subfigure}[b]{0.33\textwidth}
        \includegraphics[width=\textwidth]{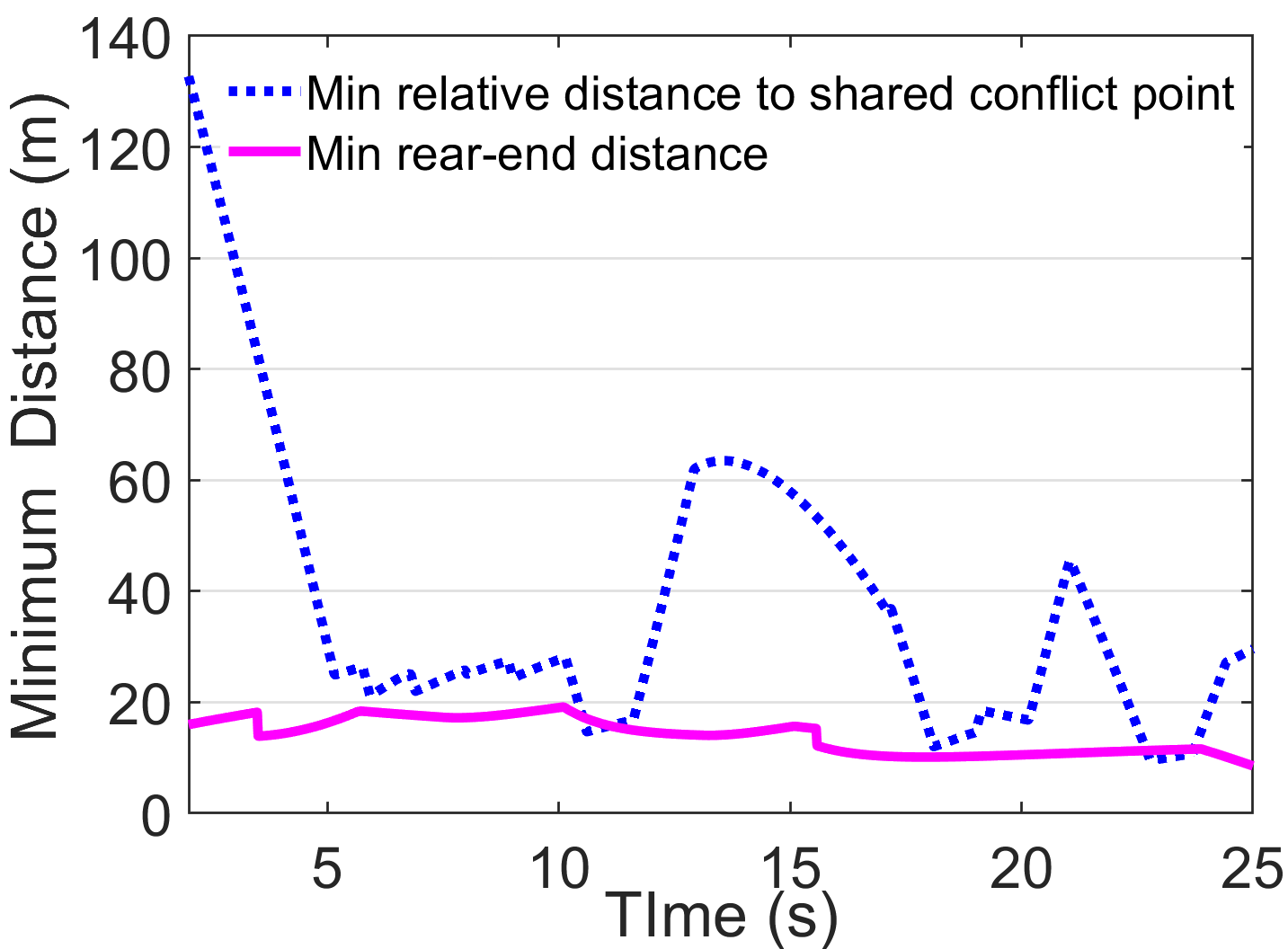}
        \caption{Minimum safety distances}
        \label{fig:Distances_scenario_2}
    \end{subfigure}
    \caption{Verification of safety between CAVs and Pedestrian in Scenario 2.}
\end{figure*}

\begin{figure}
        \centering
    \begin{subfigure}[b]{0.235\textwidth}
        \includegraphics[width=\textwidth]{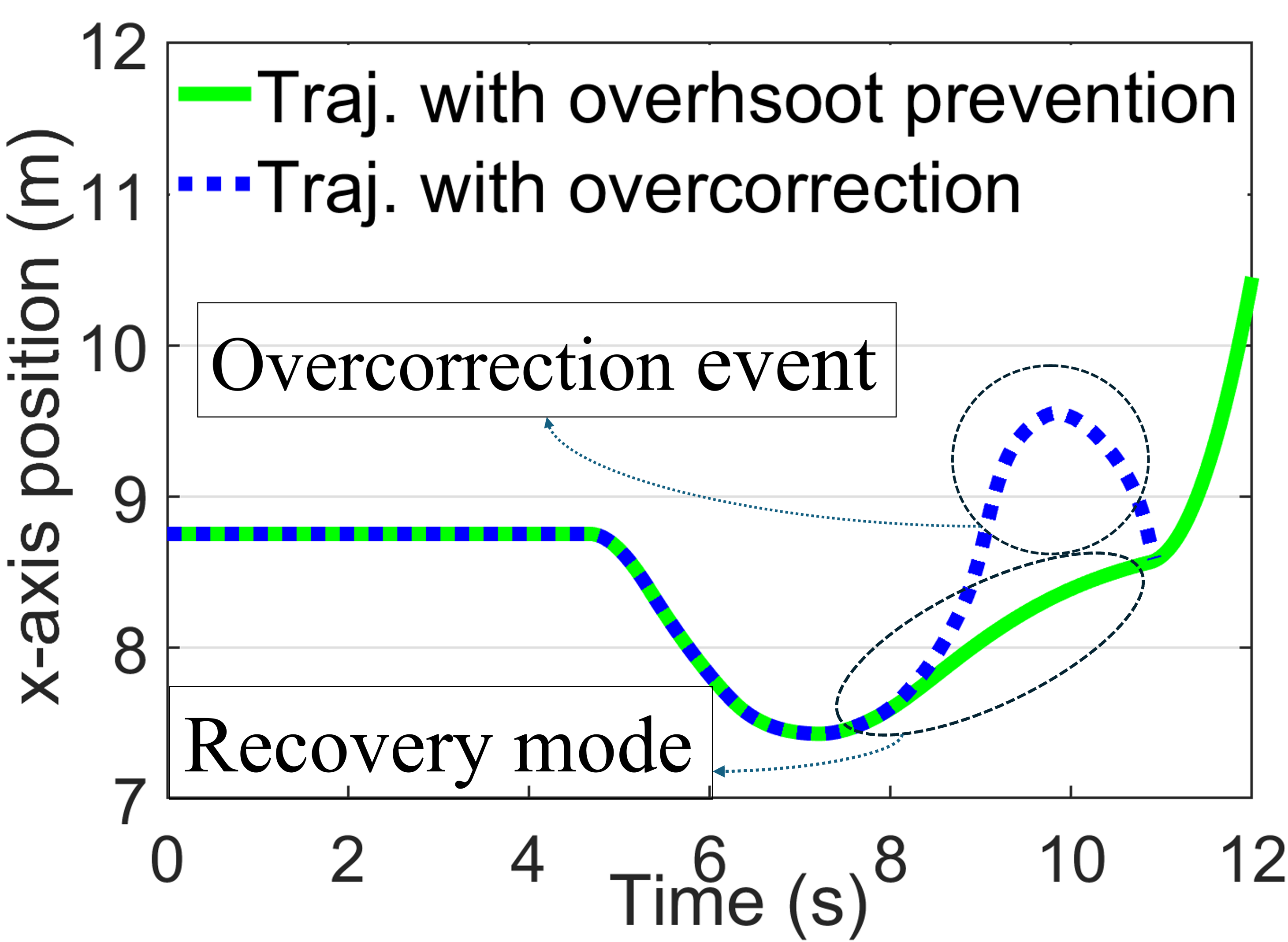}
        \caption{Trajectory of CAV 1}
        \label{fig:Trajectories}
    \end{subfigure}
    \hfill
    \begin{subfigure}[b]{0.235\textwidth}
        \includegraphics[width=\textwidth]{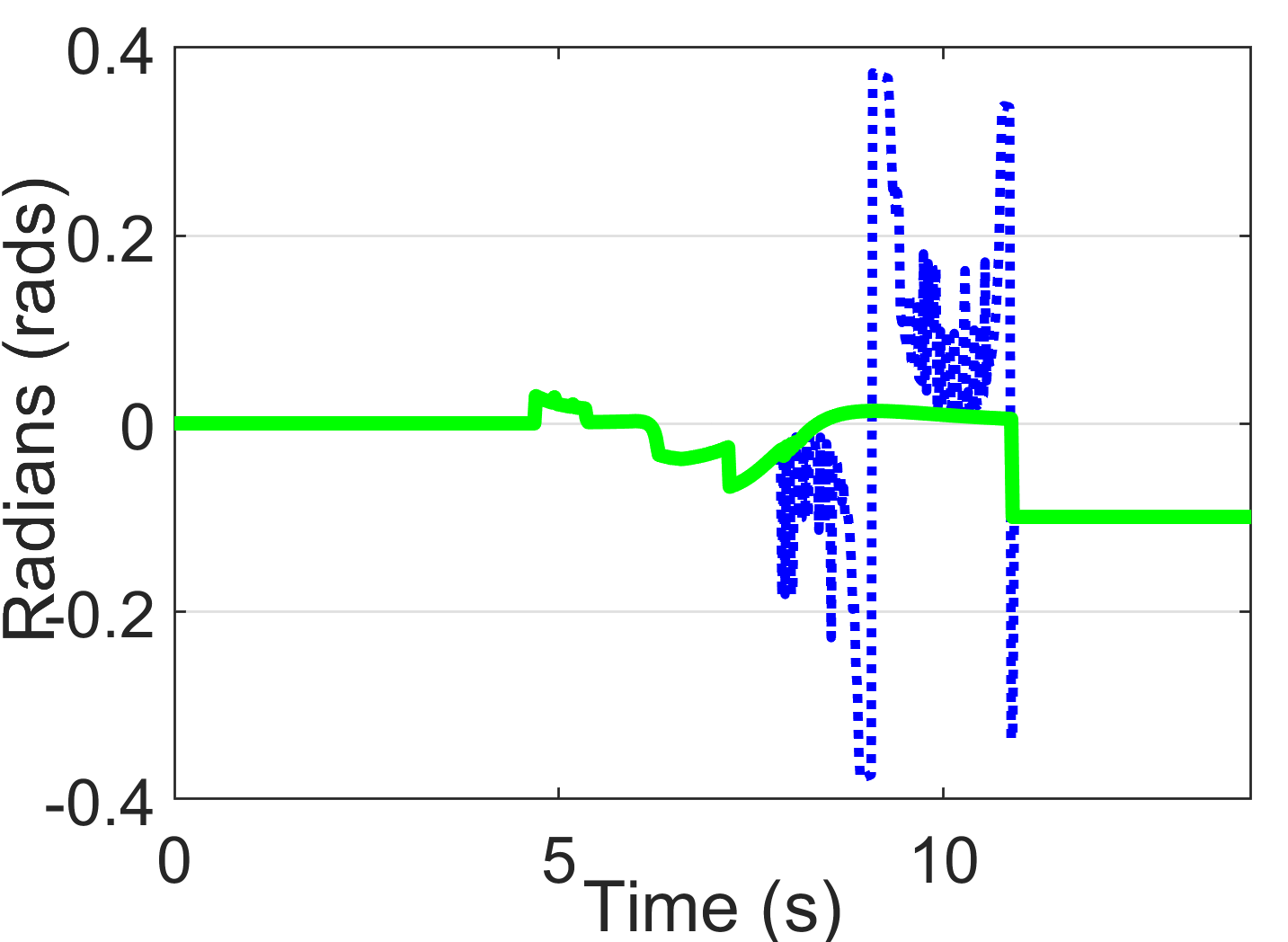}
        \caption{Steering input of CAV 1}
        \label{fig:Steering Input}
    \end{subfigure}
    \caption{Trajectory of critical vehicle and steering input in Scenario 1.}
\end{figure}

\begin{figure}
        \centering
    \begin{subfigure}[b]{0.235\textwidth}
        \includegraphics[width=\textwidth]{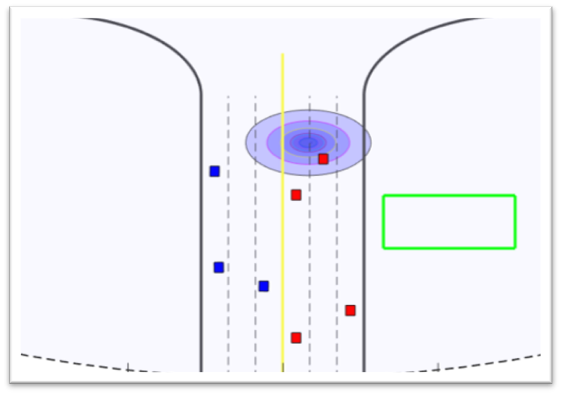}
        \caption{Snapshot 1}
        \label{fig:Snapshot_1_2}
    \end{subfigure}
    \hfill
    \begin{subfigure}[b]{0.235\textwidth}
        \includegraphics[width=\textwidth]{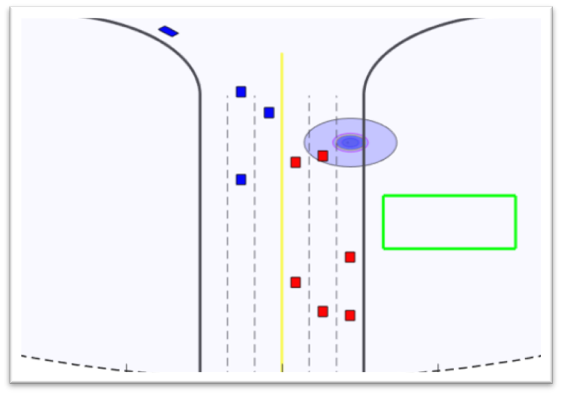}
        \caption{Snapshot 2}
        \label{fig:Snapshot_2_2}
    \end{subfigure}
    \label{fig:Snapshots_2}
    \caption{Snapshots in Scenario 2.}
\end{figure}

Regarding the incorporation of the bicycle kinematic model, we validate its usefulness in Figs. \ref{fig: snapshot 1}, \ref{fig: snapshot 2}, and \ref{fig: snapshot 3}, where CAV 1 successfully avoids the pedestrian. Finally, in Fig. \ref{fig:Trajectories}, we compare the trajectory of CAV 1 under two cases. The dotted line represents the trajectory without the inclusion of the barrier condition discussed in Remark~7, while the solid line represents the case where it is included. It is obvious that its inclusion leads to a smoother trajectory profile that gradually approaches the center of the lane without over-correction during the recovery mode. Conversely, without it, the vehicle tends to over-correct aggressively as it is presented in Fig.~\ref{fig:Steering Input} resulting in oscillations that may cause discomfort.

\textit{Scenario 2:} In this scenario, we used higher initial speeds for the CAVs and lower desired speeds for emergency mode. Additionally, we adopted a more conservative parameter selection to define the unsafe set for each CAV. In Fig. \ref{fig:Snapshot_1_2} we see a moment when several CAVs enter emergency mode due to the presence of a pedestrian. Specifically, Fig \ref{fig:Speeds_2} illustrates that at $t=7.3$ seconds, CAVs 1, 2, 3, and 4 enter emergency mode, while CAVs 5 and 6 enter the control zone a few seconds later and immediately transitioned into emergency mode, lowering their speeds to $4$ m/s. Notably, CAV 2 and 3, which were closest to the pedestrian, brake harder triggering the associated certificates while CAV 2 came even to a full stop unlike in Scenario 1 which lacked a full stop. This behavior could be attributed to either the conservativeness of the unsafe set or the pedestrian's trajectory. The safety between the CAVs and the pedestrian is verified in Fig. \ref{fig:Distances_2}. Fig. \ref{fig:Distances_scenario_2} validates the safety among CAVs. It shows both the minimum rear-end distance among all vehicles and the minimum relative distance from shared conflict points during and after the critical event. This demonstrates that the replanning mechanism effectively generates trajectories that enable all CAVs to safely navigate the intersection. Note that as it is depicted in Fig. \ref{fig:Snapshot_1_2} in this scenario, we plotted different ellipses around the pedestrian, each corresponding to a different CAV, reflecting varying levels of conservativeness among the CAVs.

We have developed a website for the paper, where we present videos from MATLAB and RoadRunner that can help visualize our results. The website can be accessed by the following link: \href{https://sites.google.com/cornell.edu/pedestrianavoidance}{https://sites.google.com/cornell.edu/pedestrianavoidance}.

\vspace{-10pt}
\section{Concluding remarks}
\vspace{-5pt}
 In this paper, we proposed a framework for signal-free intersections operated by CAVs, accounting for the presence of unexpected pedestrians. Initially, we defined the trajectories of the CAVs based on optimal control, addressing corresponding constraints either by identifying an unconstrained trajectory or by applying CBFs. Next, we relaxed the assumption presented in \cite{sabouni2024optimal}, which considers a zero standstill distance between CAVs. We then introduced our approach for collision avoidance between CAVs and pedestrians, using certificates that account for acceleration and steering angle. We also proposed a replanning mechanism that generates safe trajectories for all CAVs within the control zone. Our simulation results validated our approach, demonstrating two scenarios where CAVs successfully used both deceleration and steering to avoid a collision. Moreover, we verified that our replanning mechanism consistently returned safe trajectories for all the CAVs. Ongoing work focuses on exploring the feasibility challenges and limits posed by our approach, as well as the inclusion of HDVs.


\vspace{-10pt}
\section*{Acknowledgments}
\vspace{-5pt}
The authors would like to thank Craig Buhr and Xia Meng, MathWorks' employes, for their help in conceptualizing this paper and for contributing to the simulation‑results section, including the accompanying RoadRunner video.

\vspace{-8pt}
\linespread{0.99}\selectfont
\bibliographystyle{abbrvnat}
\bibliography{IDS_Publications_08232024,Filippos}

\begin{wrapfigure}{l}{0.8in}
    \includegraphics[width=1in,height=1.25in,clip,keepaspectratio]{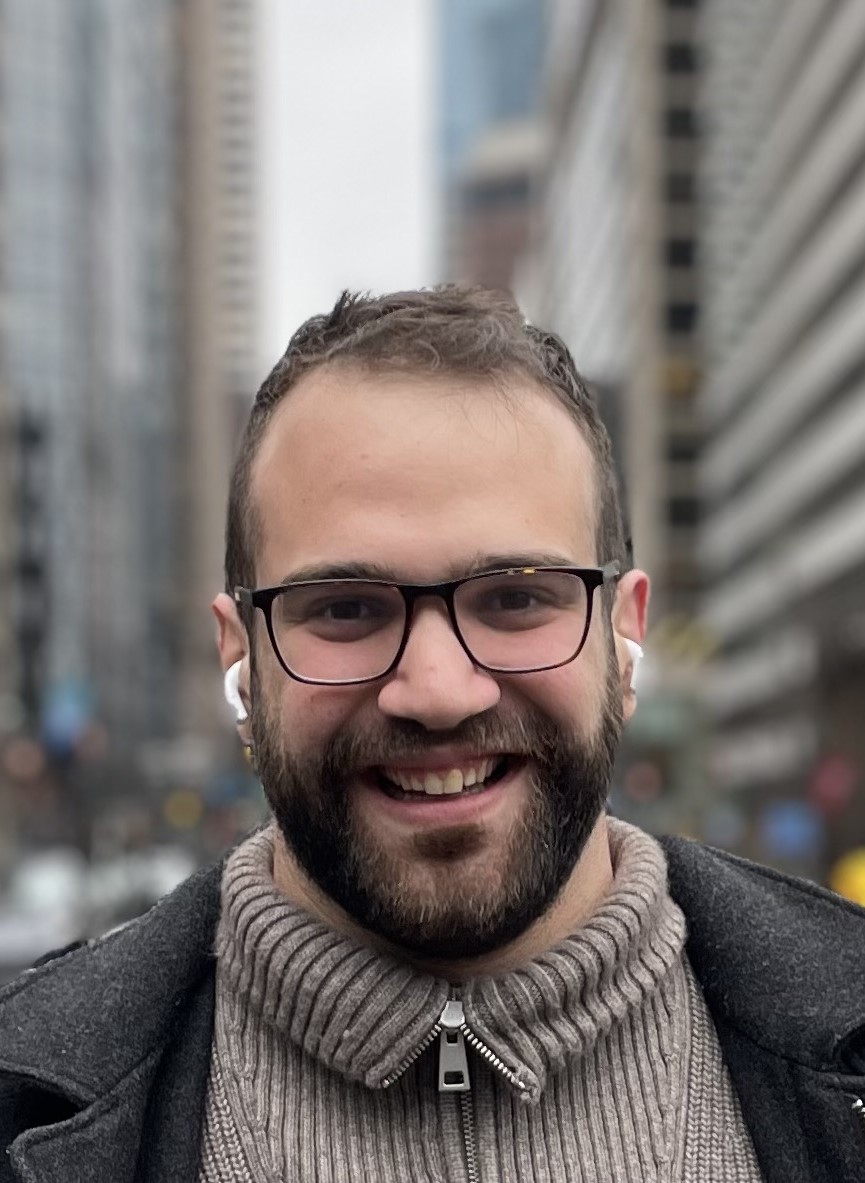}
\end{wrapfigure}

\noindent
\textbf{Filippos N. Tzortzoglou} (Student Member, IEEE) received the Diploma with an integrated M.S. degree in production engineering and management from the Technical University of Crete, Chania, Greece, in 2022. He is currently pursuing the Ph.D. degree with the Civil and Environmental Engineering Department, Cornell University, Ithaca, NY, USA. In 2022, he joined the Mechanical Engineering Department, University of Delaware, Newark, DE, USA, as a Research and Teaching Assistant. Also, in 2024, he joined MathWorks, Natick, MA, USA, as a Research Intern. His research interests lie in the area of automatic control, with applications in transportation and autonomous vehicles. Mr. Tzortzoglou has received several fellowships from foundations across the USA and the National Research Foundation.

\begin{wrapfigure}{l}{0.8in}
    \includegraphics[width=1.25in,height=1.35in,clip,keepaspectratio]{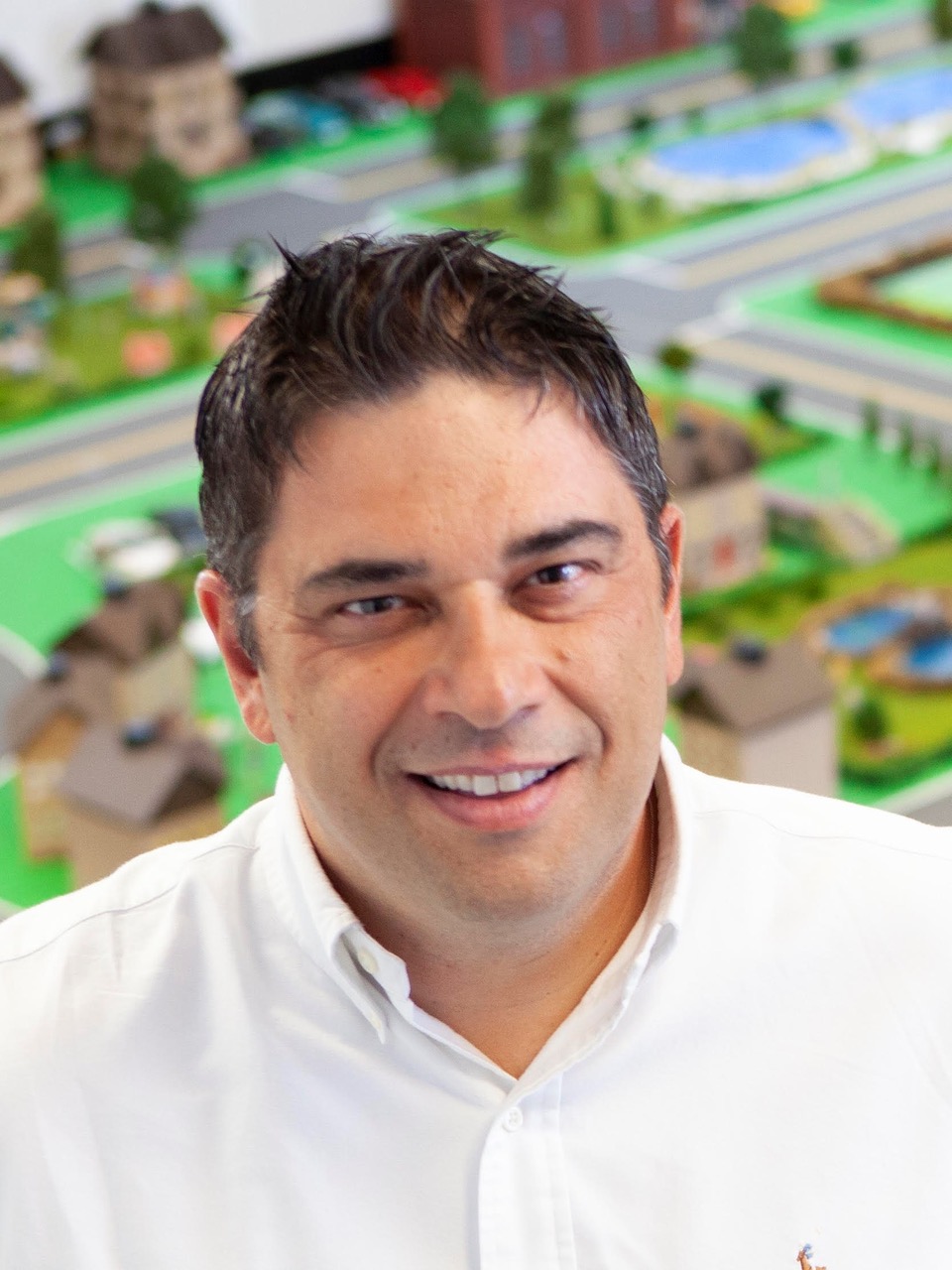}
\end{wrapfigure}

\noindent
\textbf{Andreas A. Malikopoulos} (S’06–M’09–SM’17) received a Diploma in mechanical engineering from the National Technical University of Athens (NTUA), Greece, in 2000. He received M.S. and Ph.D. degrees in mechanical engineering at the University of Michigan, Ann Arbor, Michigan, USA, in 2004 and 2008, respectively. He is a professor at the School of Civil and Environmental Engineering at Cornell University and the director of the Information and Decision Science (IDS) Laboratory. Prior to these appointments, he was the Terri Connor Kelly and John Kelly Career Development Professor in the Department of Mechanical Engineering (2017--2023) and the founding Director of the Sociotechnical Systems Center (2019--2023) at the University of Delaware (UD). Before he joined UD, he was the Alvin M. Weinberg Fellow (2010--2017) in the Energy \& Transportation Science Division at Oak Ridge National Laboratory (ORNL), the Deputy Director of the Urban Dynamics Institute (2014--2017) at ORNL, and a Senior Researcher in General Motors Global Research \& Development (2008--2010). His research spans several fields, including analysis, optimization, and control of cyber-physical systems (CPS); decentralized stochastic systems; stochastic scheduling and resource allocation; and learning in complex systems. His research aims to develop theories and data-driven system approaches at the intersection of learning and control for making CPS able to realize their optimal operation while interacting with their environment. He has been an Associate Editor of the IEEE Transactions on Intelligent Vehicles and IEEE Transactions on Intelligent Transportation Systems from 2017 through 2020. He is currently an Associate Editor of *Automatica* and *IEEE Transactions on Automatic Control*, and a Senior Editor of *IEEE Transactions on Intelligent Transportation Systems*. He is a member of SIAM, AAAS, and a Fellow of the ASME.

\end{document}